\documentclass{llncs}

\usepackage{amssymb}
\usepackage{amsmath}
\usepackage{graphicx}

\usepackage{array}
\usepackage{multirow}
\usepackage{threeparttable}

\usepackage{hyperref}
\usepackage{comment}

\usepackage{color}
\usepackage{subfigure}

\newcommand{\bbR}{\mathbb{R}}

\newcommand{\true}{\mathtt{tt}}

\newcommand{\until}[2]{\mathbf{U}_{[#1,#2]}}
\newcommand{\eventually}[2]{\mathbf{F}_{[#1,#2]}}
\newcommand{\always}[2]{\mathbf{G}_{[#1,#2]}}

\renewcommand{\phi}{\varphi}

\begin{document}

\mainmatter

\title{Smoothed Model Checking for Uncertain Continuous Time Markov Chains
\thanks{We thank Botond Cseke for sharing with us his code for Expectation Propagation,  Alberto Policriti for  highlighting the analogy with Smoothed Complexity Analysis, and Jane Hillston and Vincent Danos for useful conversations. D.M. and G.S. acknowledge support from the ERC under grant MLCS 306999. L.B. acknowledges partial support from EU-FET project QUANTICOL (nr. 600708), by FRA-UniTS, and  by the German Research Council (DFG) as part of the Cluster of Excellence on Multimodal Computing and Interaction at Saarland University.}}
\author{Luca Bortolussi\inst1\inst2\inst3 \and Dimitrios Milios\inst4 \and Guido Sanguinetti\inst4\inst5 }
\institute{Modelling and Simulation Group, University of Saarland, Germany \and Department of Mathematics and  Geosciences, University of Trieste \and 
CNR/ISTI, Pisa, Italy \and
School of Informatics, University of Edinburgh \and SynthSys, Centre for Synthetic and Systems Biology, University of Edinburgh }

\maketitle

\begin{abstract}
We consider the problem of computing the satisfaction probability of a formula for stochastic models with parametric uncertainty. We show that this satisfaction probability is a smooth function of the model parameters. This enables us to devise a novel Bayesian statistical algorithm which performs model checking simultaneously for all values of the model parameters from observations of truth values of the formula over individual runs of the model at isolated parameter values. This is achieved by exploiting the smoothness of the satisfaction function: by modelling explicitly correlations through a prior distribution over a space of smooth functions (a Gaussian Process), we can condition on observations at individual parameter values to construct an analytical approximation of the function itself. Extensive experiments on non-trivial case studies show that the approach is accurate and several orders of magnitude faster than naive parameter exploration with standard statistical model checking methods.
\end{abstract}

\section{Introduction}
Computational verification of logical properties by model checking is one of the great success stories of Theoretical Computer Science, with profound practical implications. Robust and mature tools such as PRISM \cite{Kwiatkowska2004} implement probabilistic model checking, enabling the quantification of the truth probability of a formula for a wide variety of stochastic models and logics. In many cases, however, analytical quantification is impossible for computational reasons; nevertheless, the possibility of drawing large number of samples from generative models such as Continuous Time Markov Chains (CTMCs) has led to the deployment of statistical tools. Statistical model checking (SMC, \cite{Younes2006,Zuliani2010}) repeatedly draws independent samples (runs) from the model to estimate satisfaction probabilities as averages of satisfactions on individual runs; by the law of large numbers, these averages will converge to the true probability in the limit when the sample size becomes large, and general asymptotic results permit to bound (probabilistically) the estimation error.

Both analytical and statistical tools for model checking however start from the premises that the underlying mathematical model is fully specified (or at least that a mechanism to draw independent and identically distributed samples exists in the case of SMC). This is both conceptually and practically problematic: models are abstractions of reality informed by domain expertise. Condensing the domain expertise in a single vector of parameter values is at best an approximation. While in some domains with advanced standardisation and high precision (e.g. microelectronics) this may be an acceptable approximation, in other fields (e.g. systems biology) parametric uncertainty can be very considerable. While parameter estimation from measurements has seen considerable progress in recent years \cite{opper2007variational,Andreychenko2012,Bortolussi2013qest,Georgoulas2013}, uncertainty can never be completely eliminated.  In these application areas, it would seem that acceptance of the inherent uncertainty is the natural way forward, and that alternative semantics such as Interval Markov Chains \cite{Katoen2007} or Constraints Markov Chains \cite{caillaud2011constraint} should be preferred. Model checking methodologies for these semantics are however in their infancy and mostly rely on a reduction to continuous-time Markov Decision Processes \cite{Baier2005}, obtaining upper and lower bounds on the satisfaction probability \cite{Katoen2007,benedikt_ltl_2013}, or on exhaustive (and computationally intensive) exploration of the parameter space \cite{Brim2013,milan14}.

In this paper we define a novel, quantitative approach to model checking uncertain CTMCs. We start by defining the {\it satisfaction function} of a formula, the natural extension of the concept of satisfaction probability of a formula to the case of CTMCs with parametric uncertainty. We prove that, under mild conditions, the satisfaction function is a smooth function of the uncertain parameters of the CTMC. We then propose a novel statistical model checking approach which leverages this smoothness to transfer information on the satisfaction of the formula at nearby values of the uncertain parameters. We show that the satisfaction function can be approximated arbitrarily well by a sample from a {\it Gaussian Process} (GP), a non-parametric distribution over spaces of functions, and use the GP approach to obtain an analytical approximation to the satisfaction function. This enables us to predict the value (and related uncertainty) of the satisfaction probability {\it at all values of the uncertain parameters} from individual model simulations at a finite (and generally rather small) number of distinct parameter values. We term this whole approach {\it smoothed model checking} in analogy with the recently proposed smoothed complexity analysis, another traditional domain of discrete mathematics where embedding problems in a continuous framework has proved highly valuable \cite{Spielman2009}.

The rest of the paper is organised as follows: in the next section we introduce the mathematical background on formulae and model checking, proving in the following section that the satisfaction function is a smooth function of model parameters. 
We then introduce  our smoothed model checking framework and discuss the statistical tools needed.
We demonstrate the power of our approach on a number of non-trivial examples, and conclude the paper by discussing the broader implications of our approach.

\section{Background}

In this section, we briefly review Continuous Time Markov Chains (CTMC),  Metric interval Temporal Logic, and statistical model checking approaches. 

\subsection{Continuous Time Markov Chains}
\label{sec:CTMC}

A Continuous time Markov Chain (CTMC) $\mathcal{M}$ is a Markovian (i.e. memoryless) stochastic process defined on a finite or countable state space $S$ and evolving in continuous time \cite{Durrett2012}.  We will specifically consider population models of interacting agents \cite{Bortolussi2013}, which can be easily represented by 
\begin{itemize}
\item a vector of population variables $\vec{X}=(X_1,\ldots,X_n)$, counting the number of entities of each kind, and taking values in $S\subseteq\mathbb{N}^n$;
\item a finite set of reaction rules, describing how the system state can change. Each rule $\eta$ is a tuple $\eta = (\vec{r}_\eta,\vec{s}_\eta,\tau_\eta)$.  $\vec{r}_\eta$ (respectively $\vec{s}_\eta$) is a vector encoding how many agents are consumed (respectively produced) in the reaction, so that $\vec{v}_\eta = \vec{s}_\eta - \vec{r}_\eta$ gives the net change of agents due to the reaction. $\tau_\eta = \tau_\eta(\vec{X},\theta)$ is the rate function, associating to each reaction the rate of an exponential distribution, depending on the global state of the model and on a $d$ dimensional vector of \emph{model parameters},  $\theta$. Reaction rules are easily visualised in the chemical reaction style, as 
\[r_1 X_1 + \ldots r_n X_n \xrightarrow{\tau(\vec{X},\theta)} s_1 X_1 + \ldots s_n X_n\]
\end{itemize}
The dependency on $\theta$ is often crucial, with qualitatively different dynamics arising at different values of $\theta$ (for an example, we refer the reader to the two possible regimes for infection models described in section \ref{SIR}). We use the notation $\mathcal{M}_\theta$ to stress the dependency of $\mathcal{M}$ on parameters $\theta$.

Sample trajectories of a CTMC are (usually integer-valued) piecewise constant cadlag functions\footnote{A function $f:[0,T]\rightarrow \bbR^k$ is cadlag if it is right continuous and has left limits. The space of cadlag functions can be metrised by the Skorokhod distance, making it a complete and separable metric space, see e.g. \cite{billingsley_probability_2012}  for details.}, with jumps distributed exponentially in time.
Hence, we can think of CTMCs as a collection of random variables $\vec{X}(t)$ on the state space $S$, indexed by time $t\in [0,T]$, or as random functions $\vec{X}_{0:T}$ from $[0,T]$ to $S$ (which are necessarily piecewise constant due to the countable nature of $S$). In this second sense, a CTMC is equivalent to a measure on the (infinite dimensional) space of trajectories of the system, individually denoted by $\vec{x}_{0:T}$.


\subsection{Metric interval Temporal Logic}
\label{sec:MITL}
We will specify properties of CTMC trajectories by Metric interval Temporal Logic (MiTL),  see \cite{MC:Maler:FORMATS2004:MiTL}. 
MiTL is a linear temporal logic for continuous time trajectories, and we will consider its time-bounded fragment. The choice of MiTL is justified because time-bounded linear time properties are natural when reasoning about  systems like biological ones, 
yet our method can be straightforwardly applied to any other time-bounded linear time specification formalism equipped with a monitoring routine. 

Formally, the syntax of MiTL is given by the following grammar:
\[ \phi ::= \mathtt{tt}~|~\mu~|~\neg\phi~|~\phi_1\wedge\phi_2~|~\phi_1\until{T_1}{T_2}\phi_2,   \]
where $\true$ is the true formula, conjunction and negation are the standard boolean connectives, and there is only one temporal modality, the time-bounded until $\until{T_1}{T_2}$.  Atomic propositions  $\mu$ 
are (non-linear) inequalities  on population variables. 
A MiTL formula is interpreted over a real valued function of time $\vec{x}$, and its satisfaction relation is  given in a standard way, see e.g. \cite{MC:Maler:FORMATS2004:MiTL}. 
For instance, the rule for the temporal modality states that $\vec{x},t \models \phi_1\until{T_1}{T_2}\phi_2$ if and only if $\exists t_1 \in [t+T_1,t+T_2]$ such that $\vec{x},t_1 \models \phi_2$ and $\forall t_0\in [t,t_1]$, $\vec{x},t_0 \models \phi_1$.
Temporal modalities like time-bounded eventually and always can be defined in the usual way from the until operator: $\eventually{T_1}{T_2}\phi \equiv \true\until{T_1}{T_2} \phi$ and $\always{T_1}{T_2}\phi\equiv \neg \eventually{T_1}{T_2}\neg\phi$. 
MiTL can be interpreted in the probabilistic setting \cite{Zuliani2010,Chen2011,Bortolussi2013qest} by computing the path probability $ Pr(\phi=\true\vert\mathcal{M}_\theta)$ of a formula $\phi$,  
$ Pr(\phi=\true\vert\mathcal{M}_\theta) = Pr\left(\{\vec x_{0:T}~\vert~\vec x_{0:T},0\models \phi\} \vert \mathcal{M}_\theta \right)$, i.e.\  the probability of the set of (time-bounded) CTMC trajectories that satisfy the formula\footnote{We assume implicitly that $T$ is sufficiently large so that the truth of $\phi$ at time 0 can always be established from $\vec x$. The minimum of such times can be easily deduced from the formula $\phi$ itself, see \cite{Zuliani2010,MC:Maler:FORMATS2004:MiTL}}. 
%
%
%

\subsection{Model checking}\label{checkingSection}

Model checking algorithms for MiTL formulae against a CTMC model are essentially of two kinds. Numerical algorithms \cite{Chen2011} are very complex, and severely suffer from state space explosion. A more feasible alternative for population models is Statistical Model Checking (SMC \cite{Zuliani2010,Younes2006}), which is implemented in widely used model checking tools such as PRISM \cite{kwiatkowska_prism_2011} or MRMC \cite{katoen_markov_2005}.  SMC approaches  estimate the probability of a MiTL formula by combining simulation and statistical inference tools. More precisely, given a CTMC $\mathcal{M}_\theta$ with fixed parameters $\theta$, time-bounded CTMC trajectories are sampled by standard simulation algorithms, like SSA \cite{Gillespie1977}, and monitoring algorithms for MiTL \cite{MC:Maler:FORMATS2004:MiTL} are used to asses if the formula $\phi$ is satisfied for each sampled trajectory. In this way, one generates samples from a Bernoulli random variable  $Z_\phi$, equal to 1 if and only if $\phi$ is true.
SMC then uses standard statistical tools, either frequentist \cite{Younes2006} or Bayesian \cite{Zuliani2010}, to estimate the satisfaction probability $Pr(\phi\vert\mathcal{M}_\theta)$ or to test if  $P(\phi\vert\mathcal{M}_\theta) > q$ with a given confidence level $\alpha$.

\section{Uncertain CTMCs and the Satisfaction Function}
\label{sec:problem}

In this section, we define of the main object of our study,  the satisfaction function of a formula for uncertain CTMCs, and characterize its differentiability.  As we argued in the introduction, in many fields of application where CTMCs are used, the assumption of full parametric specification is untenable; rather, the natural object to consider is a family  $ \mathcal{M}_{\theta}$ of CTMCs, where the parameters $\theta$ vary in a domain $D$. We call such a family of CTMC models, indexed by $\theta\in D$, an \emph{uncertain CTMC}. Our interest is to quantify how the satisfaction of MiTL formulae against CTMCs drawn from an uncertain CTMC depends on the unknown parameters.
The following defines the central object of study in our work.
\begin{definition}
Let $\mathcal{M}_{\theta}$ be a family of CTMCs indexed by the variable $\theta\in D$ where $D$ is a compact subset of $\mathbb{R}^d$, and let $\phi$ be a formula in a suitable temporal logic (e.g. MITL). The \emph{ satisfaction function} $f_{\phi}\colon D\rightarrow[0,1]$ associated with $\phi$ is 
\[ f_{\phi}({\theta})=Pr(\phi=\true \vert \mathcal{M}_{{\theta}})\]
i.e., with each value $\theta$ in the space of parameters $D$ it associates the satisfaction probability of $\phi$ for the model with that parameter value.
\end{definition}
Before  proceeding to show that $f_\phi(\theta)$ is a smooth function of the model parameters, we observe that classic SMC can be applied only to models with a fixed value of parameters.  Estimating the whole satisfaction function $f_{\phi}$ by SMC would require a potentially large number of evaluations; while these can be performed in parallel, the overall number of simulations needed for accurate estimation would certainly be very large.
%

\subsection{Smoothness of the  satisfaction function}

The following lemma is standard but useful (see for instance \cite{Zuliani2010}). Recall that a CTMC $\mathcal{M}_\theta$ induces a measure $\mu_\theta$ over the space of trajectories of the system.
\begin{lemma}
Let $\mathcal{M}_\theta$ be a CTMC and $\phi$ be a MiTL formula. The subset of trajectories where the formula is satisfied is a measurable set under $\mu_\theta$.\end{lemma}

We are now ready to prove our main theoretical results.
An important characterisation of the satisfaction function is given in the following theorem:
\begin{theorem}
\label{smoothnessProof}
Let $\phi$ be a MiTL formula and let $\mathcal{M}_{\theta}$ be a family of CTMCs indexed by the variable $\theta\in D$. Denote as $\boldsymbol{\tau}(\vec{X},\theta)$ the transition rates of the CTMCs and assume that these depend smoothly on the parameters $\theta$ and polynomially on the state vector of the system $\vec{X}$. Then, the satisfaction function of $\phi$ is a smooth function of the parameters, $f_{\phi}\in\mathcal{C}^{\infty}(D)$.\label{smoothness}
\end{theorem}
\begin{proof}
We begin the proof by elucidating the topology of the space of trajectories of the system.
\paragraph{Structure of space of trajectories} Let $\mathcal{T}$ denote the space of trajectories of the system starting at $t=0$ and ending at $t=T$. W.l.o.g. assume that from any state of the system there are at most $R$ different type of transitions that can occur (in a biochemical example, these would be all the possible reactions in the system). We notice that each trajectory in $[0,T]$ can be uniquely defined by specifying the number of transitions $k$, the sequence of types of transitions $r_i\quad i=1,\ldots,k$ and the times of the transitions $t_i,\quad i=1,\ldots,k$. This implies that $\mathcal{T}$ is a countable union of finite dimensional spaces $\mathcal{T}_k$ corresponding to trajectories where exactly $k$ transitions have occurred. Each of the $\mathcal{T}_k$ is in itself the union of $R^k$ copies of $[0,T]^k$; each of these copies corresponds to a sequence of $k$ transitions (there are $R^k$ such sequences), and a point in $[0,T]^k$ determines the times at which the transitions happened. Notice that, given a trajectory of the system, i.e. a point in $\mathcal{T}$, the formula $\phi$ is either true or false; this implies that the {\it set} of trajectories which satisfy the formula $\phi$ is independent of the model parameters, hence we only need to show that the measure on the set of trajectories depends smoothly on the parameters.
\paragraph{Measure on the set of trajectories - state-independent rates.} The topology of the space of trajectories corresponds to the well known factorisation of the measure over the space of trajectories; denoting a trajectory $\tau=\left[(r_1,t_1),\ldots,(r_k,t_k)\right]$, we have that \[
p(\tau)=p(k)p(r_1,\ldots,r_k\vert k)p(t_1,\ldots,t_k\vert r_1,\ldots,r_k).\]
In the case of state-independent rates, $p(k)$ is a Poisson distribution over the number of transitions (with mean $\mu$ given by the inverse of the sum of the rates times $T$), $p(r_1,\ldots,r_k\vert k)$ is the probability of the choice of the $k$ transition types (this is a rational function with positive denominator in the rates) and $p(t_1,\ldots,t_k\vert r_1,\ldots,r_k)$ is the product of the exponential probabilities of the $k$ waiting times. This density is clearly smooth w.r.t. the parameters (rates); to prove the integrability of its derivative, we have to show that the absolute value of the derivative is bound by a quantity which is integrable. As the domains $[0,T]$ and $D$ are bounded, all we have to verify is that the derivatives do not grow too fast as $k\rightarrow\infty$ to ensure integrability. This is easily verified directly for the three terms in the derivative:\begin{enumerate}
\item$\frac{\partial p(k)}{\partial \theta}=p(k-1)\frac{\partial \mu}{\partial \theta}$ where we exploited the fact that the derivative of a Poisson probability of $k$ events w.r.t. the mean is equal to the probability of $k-1$ events.
\item$\frac{\partial p(r_1,\ldots,r_k\vert k)}{\partial \theta}<kM$ where $M$ is a constant, as $p(r_1,\ldots,r_k\vert k)$ is a rational function with positive denominator and both numerator and denominator polynomials of degree at most $k$ in each of the rates.
\item$\frac{\partial p(t_1,\ldots,t_k\vert r_1,\ldots,r_k)}{\partial \theta}<kLp(t_1,\ldots,t_k\vert r_1,\ldots,r_k)$ where $L$ is a constant depending on the rates and on T, as $p(t_1,\ldots,t_k\vert r_1,\ldots,r_k)$ is a product of exponentials with exponents linear in the rates divided by a normalising constant which is a monomial of degree at most $k$ in the rates.
\end{enumerate} 
Therefore, in the case where the rates do not depend on the state of the system, all derivative terms grow at most linearly with $k$, which means they are still integrable once multiplied by the Poisson density $p(k)$.
\paragraph{Measure on the set of trajectories - state-dependent rates.} In the case where the rates depend polynomially on the state of the system, the argument above can be modified in a straightforward manner to show that each derivative term grows at most as $k^{c+1}$, where $c$ is the maximum polynomial order of the rates (w.r.t. the state variables). As a polynomial function is still integrable when multiplied by the Poisson density $p(k)$, by the same argument we obtain that the derivative of the density is still integrable over the space of trajectories, hence the first derivative of the satisfaction function exists.

We notice that repeated applications of the derivative operator still result in polynomial growth of the derivatives of the density with the number of transitions $k$. By the same argument, all the derivatives of finite order of the satisfaction function exist, hence $f_{\phi}\in\mathcal{C}^{\infty}(D)$.
QED.
\end{proof}

\section{Smoothed Model Checking}
In this section, we introduce in detail the smoothed model checking approach and the statistical tools we use. We will start first from an high level description of the method in the next section, filling in the statistical details in the following ones.

\subsection{A high level view}

Given an uncertain CTMC $\mathcal{M}_{{\theta}}$ depending on a vector of parameters ${\theta}\in\mathcal{D}$ and a MITL formula $\phi$, our goal is to find a \emph{statistical estimate} of the satisfaction probability of $\phi$ as a function of $\theta$, i.e. of the satisfaction function \[f_\phi(\theta) = {P}(\phi|\mathcal{M}_\theta),\qquad \theta\in\mathcal{D}.\] 
As in all statistical model checking algorithms, our statistical estimation will be based on monitoring the  satisfaction of the formula on sample trajectories of the system. However, as our system depends on a continuous vector of parameters, we will necessarily be able to \emph{noisily} observe the function $f_\phi$ only at a few input points $\theta_1,\ldots,\theta_k$, our {\it training set}. Given such information, our task is to construct a statistical model that, for {\it any} value $\theta^*\in\mathcal{D}$, will permit us to \emph{compute efficiently} an \emph{estimate of $f_\phi(\theta^*)$} and a \emph{confidence interval} for such a  prediction. 

To find an efficient and mathematically sound solution to this problem, we will adopt a Bayesian approach; the main ingredients are as follows: 
\begin{itemize}
\item we first choose a {\it prior distribution} over a suitable \emph{function space}; this distribution must be sufficiently expressive as to contain the satisfaction function in its support;
\item we determine the functional form of the {\it likelihood}, i.e. how the probability of the observed satisfaction values at individual parameters depends on the (unknown) true value of the satisfaction probability at that point;
\item leveraging Bayes's theorem, we then compute an approximation to posterior distribution over functions, given the observations. Evaluating the statistics of the induced posterior distribution on the function values at point $\theta^*$ we obtain the desired estimate and confidence interval. 
\end{itemize}

The technical difficulties of the procedure centre around the need to construct and manipulate efficiently distributions over functions. Here we rely on Gaussian Processes (GP), a flexible and computationally tractable non-parametric  class of distributions which provide an ideal prior distribution for our problem  \cite{Rasmussen2006}. The choice of the likelihood model is instead dictated by the nature of the problem: our observations are obtained by evaluating a property $\phi$ on a single trajectory generated from a stochastic model $\mathcal{M}_\theta$. Hence our actual observations are boolean, and the probability of the observations being 1 is a Bernoulli distribution with parameter $f_\phi(\theta)$ .  To improve estimation accuracy, we will generally draw $m\ge1$ observation for each $\theta$ in the training set,  so that observations are actually drawn from a Binomial random variable $Binomial (m,f_\phi(\theta))$. Notice that, in  statistical model checking, an approximation of  $f_\phi(\theta)$ would be computed directly from observations of such a binomial variable; the accuracy of the estimate would only be guaranteed in the limit of $m\rightarrow\infty$. In a GP context, there is no need to perform this intermediate estimation: we can \emph{directly use the binomial observation model} in Bayes's theorem. In this way, we are using   the \emph{exact} statistical model of the process, converging to the true function in the limit of a large number of observations. The strength of the approach, however, is in the fact that we generally obtain good approximations with few samples $m$ per each of few input points $\theta_1,\ldots,\theta_k$, which makes the method very efficient from a computational point of view.

The statistically minded reader may have noticed that, as we are observing a set of  true/ false labels per each input point, this problem looks similar to a classification problem (with the crucial differences that in a classification problem we observe only one single label per point and we do not have an exact statistical model of the observations). This similarity enables us to leverage the vast algorithmic repertoire developed for GP classificaton; in particular, our solution to the inference of satisfaction functions uses a modified version of the Expectation Propagation algorithm, which is the state-of-the-art in GP classification \cite{Rasmussen2006}. 

In the remainder of this section we will introduce in turn each element of the algorithm, and conclude by provide practical advice on its use and implementation. 
\subsection{Bayesian inference}
In this paper, we adopt a Bayesian machine learning approach. Bayesian methods offer substantial advantages, including a principled treatment of noise and a mathematically consistent way to quantify the resulting uncertainty in model-based estimates. Bayesian methods have already been employed in statistical model checking \cite{Zuliani2010}: there, the aim was to incorporate prior beliefs about truth probabilities, as well as to regularize to better handle rare events.

The fundamental insight of Bayesian statistics is that uncertainty quantification is a two step process: given an observable, uncertain quantity $\theta$, we must first construct a probability distribution $p(\theta)$, the {\it prior distribution}, which encapsulates our beliefs about $\theta$ prior to any observations been taken. The choice of the prior distribution is a fundamental modelling step, and affords considerable flexibility: prior distributions can range from uninformative priors, to priors encapsulating basic properties of the variables (e.g. positivity, continuity) to detailed mechanistic priors informed by hard scientific knowledge. The second component of a Bayesian model consists of the {\it likelihood function} or {\it noise model}, $p(\hat{\theta}\vert\theta)$, i.e. a probabilistic model of how the actual observations $\hat{\theta}$ depend on the value of the uncertain quantity. The likelihood effectively models the noise inherent in the observation process: while in some occasions these may be reasonably modelled from knowledge of the observation process, in other cases simple parametric choices such as Gaussian are made for computational simplicity.

Once both prior distribution and likelihood are defined, the rules of probability provide a way of computing the {\it posterior probability} through the celebrated {\it Bayes' Theorem}\begin{equation}
p(\theta\vert\hat{\theta})=\frac{p(\hat{\theta}\vert\theta)p(\theta)}{\sum_{\theta}p(\hat{\theta}\vert\theta)p(\theta)}.\label{BayesThm}
\end{equation}
The posterior probability precisely quantifies the uncertainty in our quantity of interest $\theta$ resulting from both our prior beliefs and our noisy observations of its value. Unfortunately, computing the normalisation constant in equation \eqref{BayesThm} is often computationally impossible in all but the simplest cases, and much research in Machine Learning and Computational Statistics is devoted to efficient algorithms to accurately approximate posterior distributions.

Over the rest of this section, we shall discuss how the Bayesian framework can be used effectively to devise a statistical model checking algorithm for uncertain CTMCs. We separately introduce the family of prior models we employ, how this can be combined with observations of the satisfaction of formulae over individual runs of the model, and finally how an accurate and efficient approximation of the posterior distribution can be computed.
\subsection{Prior modelling - Gaussian processes}

In this paper we are interested in estimating the satisfaction function of a formula from instances of its satisfaction on individual runs at discrete parameter values. Our theoretical analysis in Theorem \ref{smoothnessProof} enabled us to conclude that the satisfaction function is a smooth function of its arguments, the model parameters: a natural choice of prior distribution over smooth functions is a {\it Gaussian Process} (GP \cite{Rasmussen2006}). Intuitively, one can realise a random function as a linear combination of basis functions with random coefficients. If we choose the coefficients (weights) to be sampled from a normal distribution, the resulting random functions are draws from a GP. Formally, the definition of a GP is as follows:
\begin{definition} A GP is a collection of random variables indexed by an input variable $x$ such that every finite dimensional marginal distribution is a multivariate normal distribution. \label{GPdef}
\end{definition}
In practice, a sample from a GP is a random function; the random vector obtained by evaluating a sample function at a finite set of points $x_1,\ldots,x_N$ is a multivariate Gaussian random variable. A GP is uniquely defined by its {\it mean} and {\it covariance} functions, denoted by $\mu(x)$ and $k(x,x')$; the mean vector (covariance matrix) of the finite dimensional marginals are given by evaluating the mean (covariance) function on every (pair of) point in the finite sample. Naturally, by subtracting the mean function to any sample function, we can always reduce ourselves to the case of {\it zero mean} GPs; in the following, we will adopt this convention and ignore the mean function.

How is the GP covariance related to the basis function description? Consider a random function $f(x)=\sum_{i\in\mathcal{I}} w_i\xi_i(x)$ where $\mathcal{I}\subset\mathbb{N}$ is a set of indices and $\xi_i(x)$ are the basis functions. By definition, \begin{equation}
k(x,x')=\langle f(x)f(x')\rangle=\sum_{i,j}\xi_i(x)\xi_j(x')\langle w_iw_j\rangle\label{GPcovBF}\end{equation}
so that the covariance function is determined by the covariance of the weights and by products of basis functions evaluated at the two points. 
In practice, it is more convenient to specify a GP by directly specifying its covariance function (the so-called function space view), rather than using the explicit basis function construction; nevertheless, basis functions are still useful in order to prove properties of GPs. 

A popular choice for the covariance function, which we will also use, is the {\it squared exponential} covariance function\[
k(x,x')=\sigma^2\exp\left[-\frac{(x-x')^2}{\lambda^2}\right]\]
with two hyper-parameters: the amplitude $\sigma^2$ and the characteristic length scale $\lambda^2$. It is easy to show that this covariance function corresponds to selecting as basis functions a set of Gaussian shaped curves of standard deviation $\lambda$ and centred at all points in the input space \cite{Rasmussen2006}.

The kernel or covariance function endows the space of samples from a GP with a metric: this is an example of a Reproducing Kernel Hilbert Space (RKHS). A complete characterisation of such spaces is non-trivial; for our purposes, however, it is sufficient to show that their expressivity is sufficient to approximate a satisfaction function by a sample from a GP. We restrict to the squared exponential covariance, although our result holds for the more general class of universal covariance functions (for a precise definition of universality see \cite{steinwart_influence_2002}). We then have the following result   


\begin{theorem}
\label{denseness}
Let $f$ be a continuous function over a compact domain $D\in\mathbb{R}^p$. For every $\epsilon>0$, there exists a sample $\psi$ from a GP with squared exponential covariance such that\[
\Vert f-\psi\Vert_2\le\epsilon\]
where $\Vert\cdot\Vert_2$ denotes the $L_2$  norm.
\end{theorem}
\begin{proof}
This follows directly from the universal property of the square exponential kernel, as discussed in \cite{steinwart_influence_2002}
\footnote{The proof of the universality of the square exponential kernel is nontrivial and technical. A simpler, constructive proof of this theorem can be obtained by observing that any smooth function can be approximated arbitrarily well with a polynomial, and then use the fact that the integral operator with squared-exponential kernel has bounded inverse on polynomials. Therefore, the weight vector must be bounded, and hence have finite probability under a GP.
}.
 QED.

%
%
%
\end{proof}

\subsection{GP posterior prediction}
The results of Theorems \ref{smoothness} and \ref{denseness} jointly imply that the satisfaction function of a formula can be approximated arbitrarily well by a sample from a GP, justifying the use of GPs as priors for the satisfaction function.
To see how this fact enables a Bayesian statistical model checking approach {\it directly at the level of the satisfaction function}, we need to explain the basics of posterior computation in GP models.  Let $x$ denote the input value and let $\mathbf{\hat{f}}=\{\hat{f}_1,\ldots,\hat{f}_N\}$ denote observations of the values of the unknown function $f$ at input points $x_1,\ldots,x_N$. We are interested in computing the distribution over $f$ at a {\it new} input point $x^*$ {\it given} the observed values $\mathbf{\hat{f}}$, $p(f(x^*\vert\mathbf{\hat{f}})$. A priori, we know that the true function values at any finite collection of input points is Gaussian distributed, hence\[
p\left(f(x^*),f(x_1),\ldots,f(x_N)\right)=\mathcal{N}(\boldsymbol{\mu},\Sigma)\]
with $\boldsymbol{\mu}$ and $\Sigma$ obtained from the mean and covariance function as explained in the previous subsection. This prior distribution can be combined with likelihood models for the observations, $p(\hat{f}\vert f)$ in a Bayesian fashion to yield a joint posterior\[
p\left(f(x^*),f(x_1),\ldots,f(x_N)\vert\mathbf{\hat{f}}\right)=\frac{1}{Z}p\left(f(x^*),f(x_1),\ldots,f(x_N)\right)\prod_ip(\hat{f}_i\vert f(x_i))\]
where $Z$ is a normalisation constant. The desired posterior predictive distribution can then be obtained by integrating out ({\it marginalising}) the true function values $f(x_1),\ldots,f(x_N)$\begin{equation}
p(f(x^*\vert\mathbf{\hat{f}})=\int \prod_{i=1}^Ndf(x_i)p\left(f(x^*),f(x_1),\ldots,f(x_N)\vert\mathbf{\hat{f}}\right).\label{predPost}\end{equation}
Equation \eqref{predPost} plays a central role in non-parametric function estimation; the inference procedure outlined above goes under the name of {\it GP regression}. It is important to note that, in the case of Gaussian observation noise, the integral in equation \eqref{predPost} can be computed in closed form. Further details are given e.g. in \cite{Rasmussen2006}.


\paragraph{{\bf Important remark}: GP regression provides an analytical expression for the predicted mean and variance of the unknown function at all input points.}
\subsection{Likelihood model}
In our case, observations of the satisfaction function are made through boolean evaluations of a formula over individual trajectories at isolated parameter values. Therefore, the Gaussian observation noise cannot be applied directly in this case, meaning that a closed form solution to the inference problem cannot be found\footnote{In principle, one could appeal to the Central Limit Theorem and approximate a Gaussian observation model by using as observations averages of formula evaluations over numerous trajectories for each values of the parameters. Such a procedure however would incur significant computational overheads.}. However, as we shall see, a closed form, accurate approximation of the posterior can be found. The satisfaction of a formula $\phi$ over a trajectory generated from a specific parameter value $\theta$ is a Bernoulli random variable with success probability $f_{\phi}(\theta)$. In order to map this probability to the real numbers, we introduce the {\it inverse probit} transformation\[
\psi(f)=g\Leftrightarrow f=\int_{\-\infty}^g\mathcal{N}(0,1)\quad \forall f\in[0,1], g\in\mathbb{R}\]
where $\mathcal{N}(0,1)$ is the standard Gaussian distribution with mean zero and variance 1. The function $g_{\phi}(\theta)=\psi(f_{\phi}(\theta)$ is by construction a smooth, real valued function of the model parameters, and can therefore be modelled as a draw from a GP. 

We can summarise the inference problem as follows: our data would consist of $D$ binary evaluations of satisfaction at each of $P$ parameter values. At each parameter value, binary evaluatsion represent independent draws from the same Bernoulli distribution with success probability $f_{\phi}(\theta)$. The (inverse probit) transform of the success probability is a smooth function of the parameters and is assigned a GP prior. The overall joint probability of the observations $\mathcal{O}$ and of the satisfaction function $f_{\phi}(\theta)$ would be given by \begin{equation}
p\left(\mathcal{O},f_{\phi}(\theta)\right)=GP\left(\psi(f_{\phi}(\theta))\right)\prod_{i=1}^D\prod_{j=1}^P\mathrm{Bernoulli}\left(O_{i,j}\vert f_{\phi}(\theta_j)\right)\label{joint}
\end{equation}
Computing the posterior distribution the satisfaction probability at a new parameter value $f_{\phi}(\theta^*)$ solves the statistical model checking problem.

Before discussing how the posterior computation problem can be solved, we make the following observations:
\begin{enumerate}
\item{The procedure described above is an {\it exact} probabilistic model, in the sense that both the Bernoulli and the GP modelling step use provable properties of the unknown function (by virtue of Theorem \ref{smoothnessProof}. No further approximations are introduced.}
\item{Computing the posterior distribution introduces an approximation, however it provides an {\it analytical estimate} (with uncertainty) of the satisfaction probability at all parameter values (not only the ones explored through simulation)}
\item{The posterior inference problem is akin to a {\it classification} problem, where we seek to assign the probability that a binary output will be 1 to a point in input space. The difference from a standard classification problem is that in this case we admit the possibility of having multiple labels observed for the same input point. This observation enables us to effectively exploit the vast amount of research on GP classification over the last decade to devise an efficient and accurate algorithms for statistical model checking.}
\end{enumerate} 
\subsection{Approximate posterior computation}
GP classification has been intensely studied over the last fifteen years in machine learning; a key difference from the regression case is that exact computation of the posterior probability is not possible. However, several highly accurate approximate schemes have been proposed over the last decade in the machine learning literature: here we use the Expectation Propagation (EP) approach \cite{Opper:Gaussian00,Minka:expectation01}, which has been shown to combine high accuracy with strong computational efficiency. Importantly, this approach still provides an analytic approximation to the whole satisfaction function in terms of basis functions.

We highlight here the basic ideas behind the EP algorithm; for a more detailed discussion see e.g. \cite{Rasmussen2006}. The EP algorithm has its origins in the statistical physics of disordered systems \cite{Opper:Gaussian00}; it is a general algorithm that computes a Gaussian approximation to probabilistic models of the form\begin{equation}
p(\mathbf{x}\vert\mathbf{y})=p_0(\mathbf{x})\prod_i t_i(y_i,x_i)\label{latentGaussian}
\end{equation}
where $p_0(\mathbf{x})$ is a multivariate Gaussian distribution coupling all $x_i$ variables (so called {\it site variables}) and $t_i$ can be general univariate distributions. Models of this form are frequently encountered in machine learning and are termed {\it latent Gaussian models}: the common setup (which is indeed our setting in equation \eqref{joint}) is that the $p_0$ term represents a prior distribution, with the $t_i$ terms representing non-Gaussian observation likelihoods. The EP approximation inherits the same structure of the original model in equation \eqref{latentGaussian}, with the likelihood terms replaced by univariate Gaussian terms\begin{equation}
q(\mathbf{x}\vert\mathbf{y})=p_0(\mathbf{x})\prod_i \tilde{t}_i(y_i,x_i).\label{EPapprox}
\end{equation}
The algorithm then proceeds iteratively as follows:\begin{enumerate}
\item{Given an estimate of the joint EP approximation in equation \eqref{EPapprox}, chose a site index $i$, remove the approximate likelihood term corresponding to site $i$, and marginalise all other variables. The resulting distribution\[
q_c(x_i)=\int \prod_{i\neq i}dx_jq(\mathbf{x}\vert\mathbf{y})\tilde{t}^{-1}(y_i,x_i)\]
is called the {\it cavity distribution}. This term, inherited from statistical physics, is intuitively interpreted as the influence site $i$ would feel from the other sites {\it if it was not there}.}
\item{The exact likelihood term corresponding to site $i$ is reintroduced forming the so-called {\it tilted distribution}, $q_c(x_i)t_i(x_i,y_i)$ (which is in general not normalised).}
\item{The EP approximation is updated by replacing the initial approximate likelihood term for site $i$ with a new term $\tilde{t}_i^{new}(x_i,y_i)$ obtained by finding the univariate Gaussian which matches the moments of the tilted distribution}
\item{This update procedure is applied iteratively to all sites until the moment of the EP approximation do not change.}
\end{enumerate}

A formal justification of the EP algorithm is non-trivial: \cite{Heskes2005} showed that the EP free energy does indeed minimise a Gibbs free energy associated with the statistical model, however their argument is non-trivial and out of the scope of this report. Importantly, extensive empirical studies have confirmed the excellent accuracy of the EP algorithm for posterior computation in GP classification models \cite{Kuss:assessing05}, so that EP is effectively the state-of-the-art approximate inference algorithm in this field.
\subsection{Practical considerations}
Smoothed model checking relies on the density of the RKHS within the space of smooth functions; this restricts somewhat the choice of covariance function that can be used to model the satisfaction function. Nevertheless, several possible choices still remain besides the squared exponential kernel we use; we refer the interested reader to \cite{Rasmussen2006}. Each of these covariance functions is equipped with some hyperparameters; while the value of the hyperparameters does not affect the theoretical guarantees (i.e. the RKHS is dense in the smooth functions regardless of the kernel hyperparameters), their value may have practical repercussions on the quality of the reconstruction from a finite sample. Automatic estimators for hyperparameters can be obtained by type-II maximum likelihood methods \cite{Rasmussen2006} and these may be very useful when the dimensionality of the space is high, so that empirical, nonparametric estimates are difficult to obtain.

An important issue is the computational complexity of Smoothed Model Checking. As we will see in Section \ref{sec:experiments}, leveraging the smoothness of the satisfaction function can dramatically reduce the number of simulations required for estimation; this however comes at a cost. GP regression involves the inversion of the covariance matrix estimated at all pairs of parameter values; if we have $n$ points in our grid of parameter values, this comes at a complexity $O(n^3)$. This complexity can be alleviated by using sparse approximations (again intensely studied in the last ten years, see \cite{Rasmussen2006} for a review). Further savings may be obtained by an intelligent choice of grid points, using strategies such as Latin Hypercube Sampling (see e.g. \cite{Bortolussi2013qest}). 

Our theoretical results Theorems 1 and 2 ensure that the statistical model we employ for the satisfaction function is correct; therefore, powerful arguments from asymptotic statistics guarantee that, when the number of samples increases, the estimated satisfaction function (GP posterior mean) will approach the true satisfaction function. Practical guidance as to exactly how many samples are necessary for a certain accuracy is harder to come by; our advice would be to monitor the predictive uncertainty which is also returned by the GP. The predictive uncertainty also diminishes as the number of samples increases, and monitoring of this quantity would enable the user to obtain the desired level of precision.

\subsection{Implementation}

A prototype Java implementation is available at \url{http://homepages.inf.ed.ac.uk/dmilios/smoothedMC}.
Our tool allows the user to perform smoothed model checking on one or more parameters via a command-line interface.
The program has to be provided with two input files: the first one contains the model specified in the Bio-PEPA language \cite{Ciocchetta2009}; the second input file contains the MiTL property to be examined.
The user has also to specify the names of the parameters to be investigated and the corresponding ranges/granularity.
The output is in the form of two \lq\lq .csv\rq\rq\ files that contain a grid of parameter values and their corresponding predictions of the (expected) satisfaction probability, along with their 95\% confidence bounds.

\section{Experiments}
\label{sec:experiments}

In this section, we present a series of experiments that demonstrate the potential of Smoothed Model Checking.

\subsection{Poisson process}
As a simple benchmark for the approach, we chose to investigate the satisfaction of the simple formula $\always{0}{1}(N<4)$ on a Poisson process $N$ with uncertain rate. While this scenario is very simple, the Poisson process remains a fundamental building block of many more complex models (such as queueing networks), and it does have the advantage that an analytic expression can be simply derived for the satisfaction function as a function of the uncertain parameter. The arrival rate of the Poisson process was allowed to vary over an order of magnitude from 0.5 to 5; within this range the satisfaction function decreases monotonically from practically 1 to about 0.3.

To estimate the satisfaction function, we selected a grid of 46 parameter values at a distance of 0.1 from each other and used a GP prior with squared exponential covariance with amplitude 1 and characteristic length scale $1$. GP hyper-parameters were not optimised in this case. Figure \ref{PoissonResults} shows the results of using the hierarchical GP model with a \emph{single} observation at each parameter value (left) and with \emph{five} observations per each parameter value (right). As is evident, both cases capture the basic shape of the satisfaction function, but the use of 5 replicate observations per input values considerably improves the situation. In order to quantify the quality and reliability of the estimation, we carried out a more thorough study, simulating sixty independent data sets with a single observation, 5 and 10 observations per input value (twenty each). Summary statistics for this larger experiment are reported in Table \ref{table:PoissonBig}. We quantified the accuracy of our estimation using the mean squared error (MSE) of our estimation across all input points (sum of the squared residuals): each row reports the minimum, mean, and maximum MSE across the twenty runs. To quantify the accuracy of the uncertainty quantification, we report the average fraction of true function values that fall outside the 95\% confidence intervals of the GP. Finally, in order to assess whether leveraging the smoothness of the satisfaction function brought concrete computational benefits, we also indicate the number of standard SMC simulations that would be needed at each input value to yield an expected MSE equal to the one we empirically find with our method (as the true satisfaction function is known, it is easy to calculate the expected SMC MSE as the sum of the theoretical variances at each point). As can be seen, the Smoothed Model Checking estimate is extremely accurate already with as few as 5 observations per parameter value. Importantly, our uncertainty quantification appears to be well calibrated: as the number of observations per parameter value increases, the fraction of true function values which lies outside the 95th percentile approaches 5 \%, as expected. Finally, we notice that smoothed MC can entail very considerable computational savings over the simpler approach of running SMC independently at each parameter value. Naturally, our approach also yields a powerful tool to predict  and quantify uncertainty on the satisfaction values at all parameter values; in order to do this from standard SMC estimation, one would need to further interpolate the function values obtained (e.g. using GP regression).

\begin{figure}[!t]
\begin{center}
\includegraphics[width=0.45\textwidth]{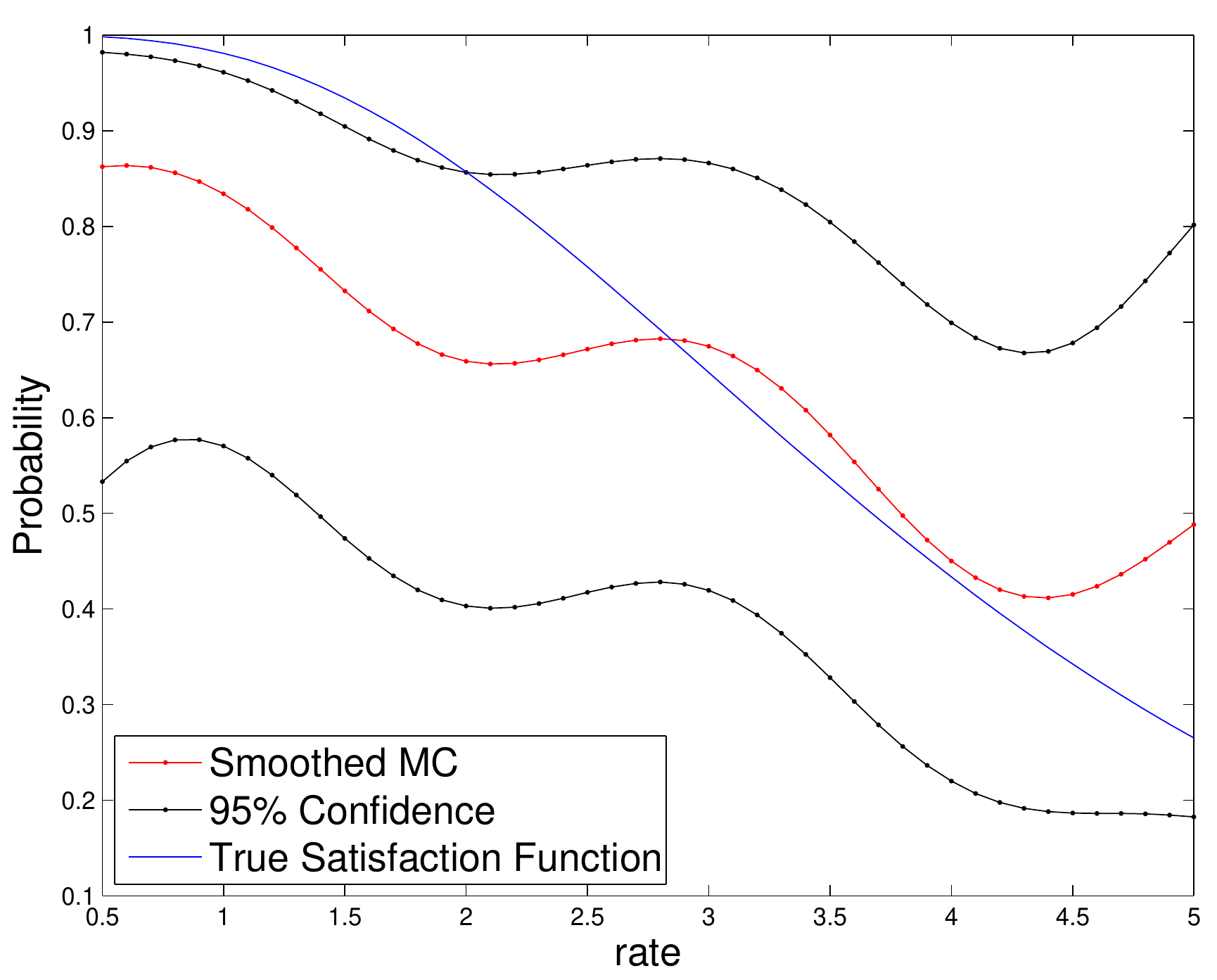}
\includegraphics[width=0.45\textwidth]{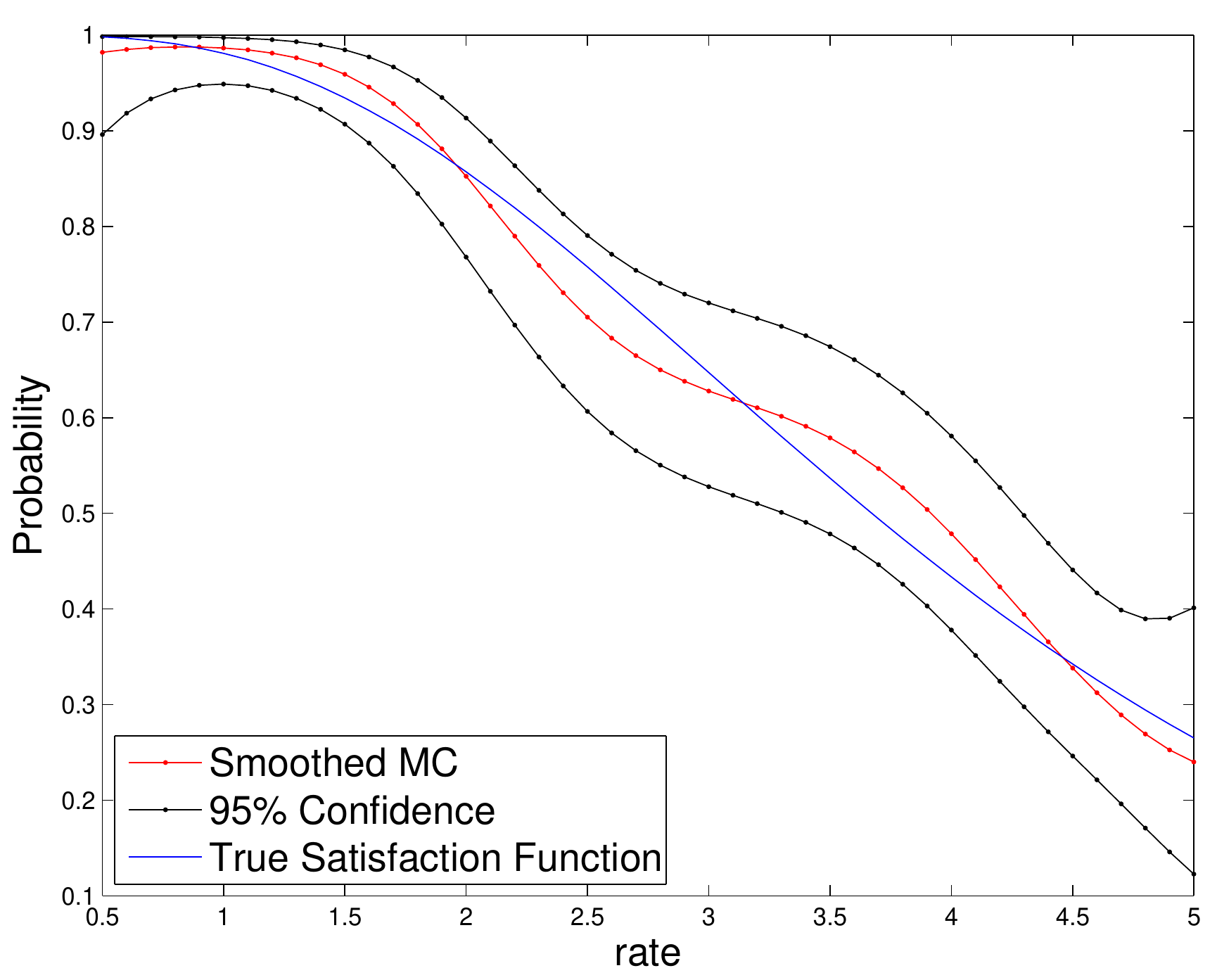}
\end{center}
\caption{Estimated satisfaction function for Poisson process from a single sample per input value (left) and five samples per value (right). Blue: true satisfaction function; red: mean estimate; black: estimated 95\% confidence intervals.}\label{PoissonResults}
\end{figure} 

\begin{table}[!t]
\[
\begin{array}{|c|c|c|c|}
\hline
\text{Obs. per value}  & \text{residuals} & \text{mean fraction outside 95 percentile} & \text{SMC}\\
\hline
1 & 0.0014\colon 0.011 \colon 0.034 & 11\% & 4\colon 11\colon 87\\
\hline
5 & 0.0011\colon 0.0064 \colon 0.014 & 6.3\% & 9\colon 20\colon 109\\
\hline
10 & 0.0009\colon 0.0037 \colon 0.0086 & 4.8\% & 15\colon 34\colon 132\\
\hline

\end{array} 
\]
\caption{Poisson process results: statistics over twenty independent runs with 1, 5 and 10 truth observations per parameter values. The second column reports the MSE in the reconstruction of the satisfaction function (min:mean:max across the twenty runs). The third column reports the average fraction of predictions that fall outside the GP's 95\% predicted uncertainty. The fourth column indicates the (theoretical) number of SMC runs per point which would be needed to match the max:mean:min empirical Smoothed model checking results.}
\label{table:PoissonBig}
\end{table}

The results on this simple example show that smoothed model checking can be a very effective tool for estimating satisfaction functions, and that the GP classification framework can yield very considerable computational savings over SMC estimation at isolated parameter values.

\subsection{Network epidemics}\label{SIR}
We consider now a more structured example of the spread of an epidemics in a population of fixed size. We will consider  the classical SIR infection model \cite{Andersson2000}, in which an initial population of susceptible nodes can be infected after contact with an infected individual. Infected nodes can recover after some time, and become immune from the infection; such models play an important role not only in epidemiology, but also in computer science, for instance  as models of the spread of software worms. Here we consider the case of permanent immunisation.

This system is modelled as a population CTMC, in which the state space is described by a vector $\vec{X}$ of three variables, counting how many nodes are in the susceptible ($X_S$), infected ($X_I$), and recovered ($X_R$). The dynamics of the CTMC are described by a list of transitions rules, or reactions, together with their rate functions. We represent them in the biochemical notation style (see e.g. \cite{Gillespie1977}). All rates of this model follow the law of mass action.  
\begin{description}
\item[Infection:]  $S+I \xrightarrow{k_i} I+I$,  with rate function  $k_i X_S X_I$;
\item[Recovery:] $I\xrightarrow{k_r} R$,  with rate function  $k_r X_I$;
\end{description}

Since immunisation is permanent, the epidemics extinguishes after a finite amount of time with probability one. However, the time of extinction depends on the parameters of the process in a non-trivial way.  As for the transient dynamics before extinction, there are two possible behaviours depending on the \emph{basic reproductive number} $R_0 = \frac{k_i}{k_r}$. If $R_0 < 1$, the epidemics extinguishes very quickly, while if $R_0>1$, there is an outbreak and a large fraction of the population can get infected. In this example, we consider the following MiTL property, which states that the extinction happens between time 100 and 120:
\begin{equation}
\label{eq:ext}
\phi =  ( \eventually{100}{120} X_I =  0 ) \wedge ( \always{0}{100} X_I > 0 ). 
\end{equation}

\begin{figure}
\begin{center}
\includegraphics[width=.48\textwidth]{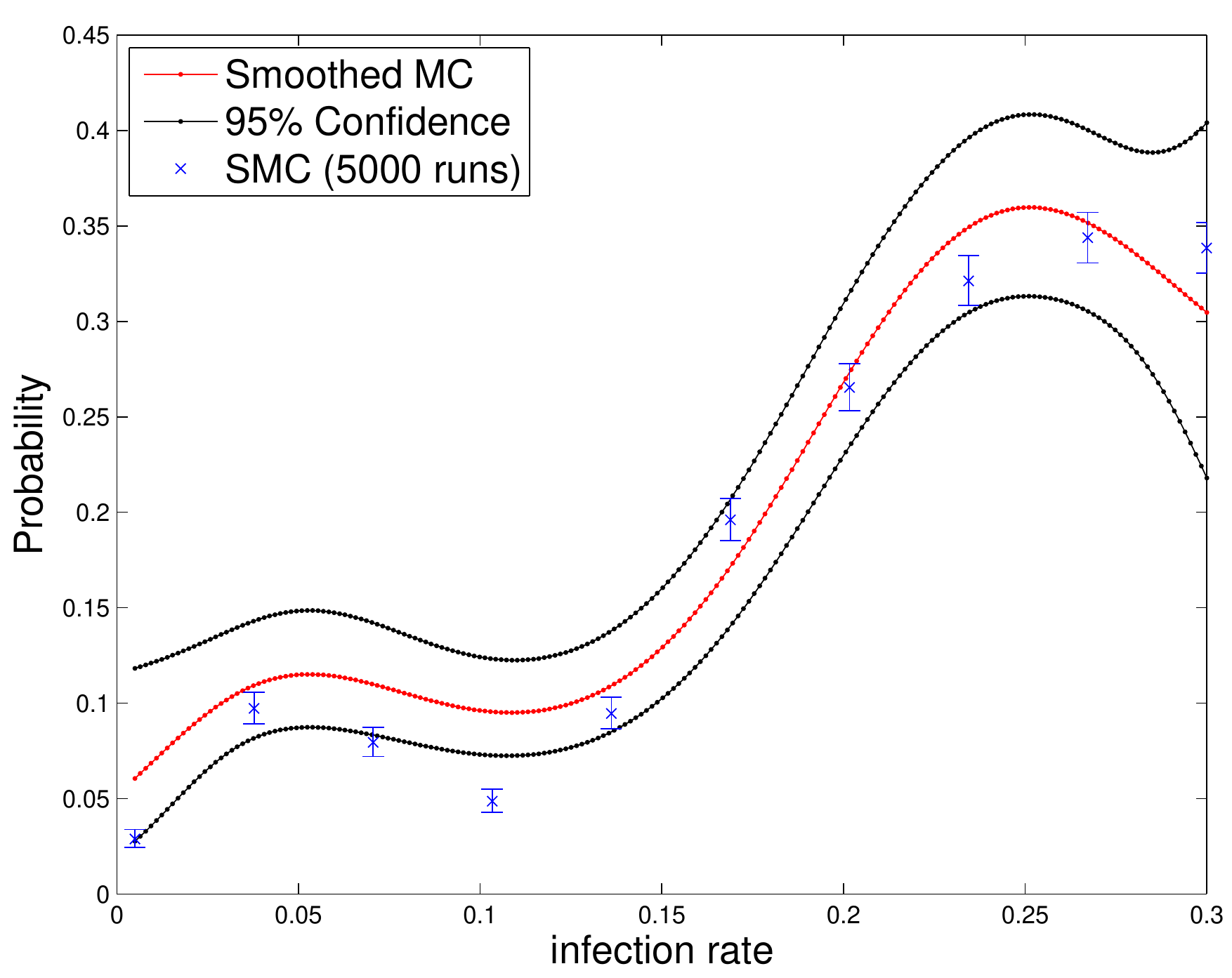} \includegraphics[width=.48\textwidth]{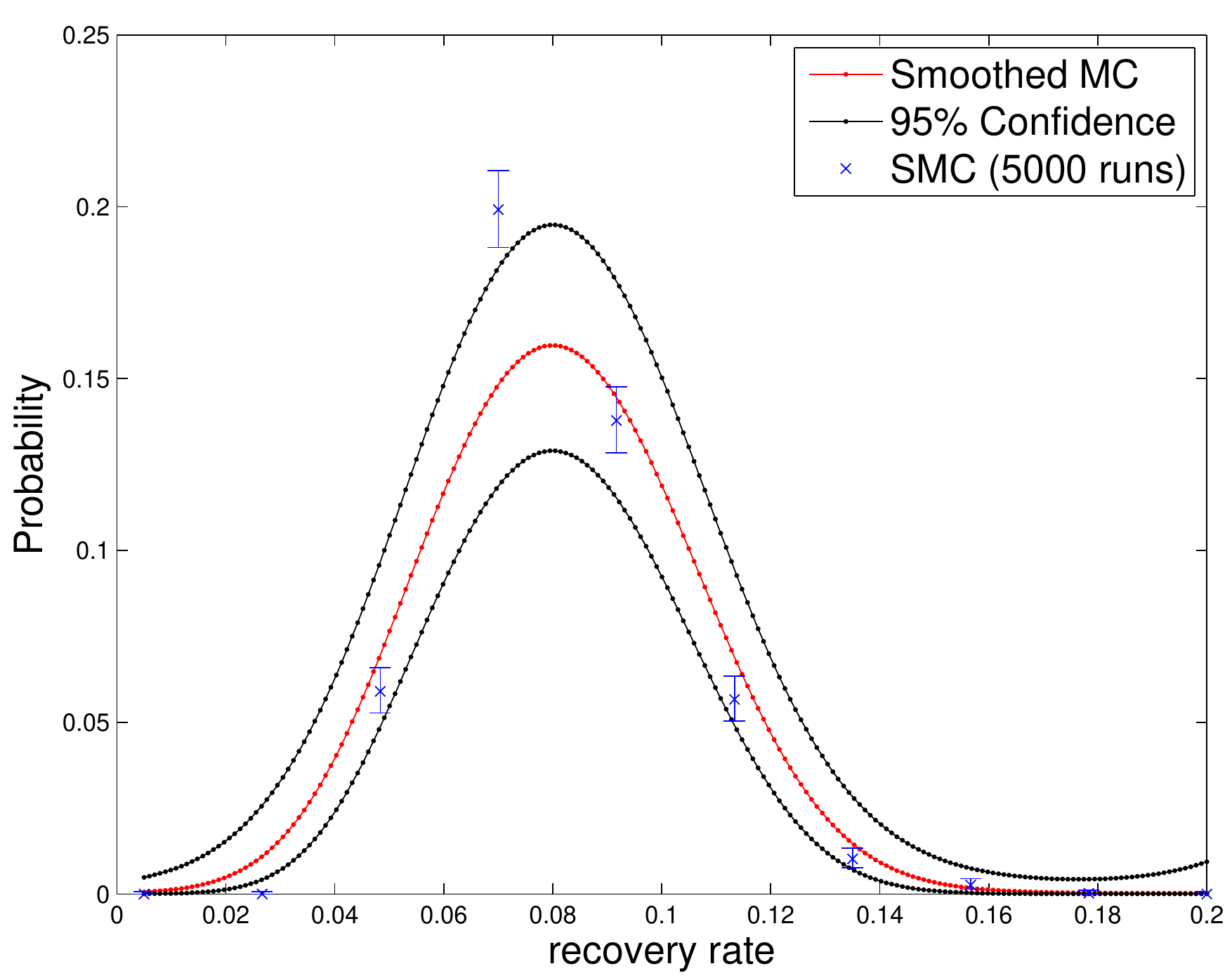}
\end{center}

\caption{Smoothed model checking reconstruction of the satisfaction probability of the extinction formula (\ref{eq:ext}), as a function of $k_i\in [0.005,0.3]$ (left) and of $k_r\in[0.005,0.2]$ (right). The red line is the reconstructed probability and the black curves are the 95\% confidence bounds. The blue crosses, instead, are the values of the satisfaction probability estimated with standard statistical model checking routines with 5000 runs. Confidence bounds are reported also for such an estimate. }
\label{fig:SIRsingle}
\end{figure}

We first used smoothed model checking to explore each of the two parameters individually. We varied the infection rate $k_i$ in the interval $k_i\in [0.005,0.3]$, holding $k_r$ fixed to $k_r = 0.05$, while $k_r$ has been varied within $[0.005,0.2]$ with $k_i=0.12$. 
In each case, we sampled 10 runs from a grid of 200 points evenly distributed in the parameter domain, for a total of 2000 runs,  and applied smoothed model checking to obtain a prediction of the satisfaction probability of $\phi$ as a function of the parameter. 
Results are shown in Figure \ref{fig:SIRsingle}, where we also plot the 95\% confidence bounds of the estimate. The predicted satisfaction function is compared with  point-wise estimates of the satisfaction probability using standard statistical model checking with 5000 runs per point.  As we can see, our method provides an accurate reconstruction of the probability function, at a very cheap computational cost, which is dominated by simulation.

As a final experiment, we considered the estimation of the satisfaction probability of $\phi$ as a function of both parameters. Results are shown in Figure \ref{fig:SIRdouble}, which compares results of Smoothed Model Checking with traditional Statistical Model Checking: the left hand plot shows the GP posterior mean having observed 10 independent runs per parameter value, while the right hand panel shows a surface plot of estimates of the probability function at the same points obtained with SMC from 5000 simulations per point (deep SMC). The remarkable similarity between the two surfaces attests to the good quality of the Smoothed Model Checking estimation; nevertheless, some of the limitations of the approach are also evident. In particular, the satisfaction function changes very steeply as the recovery rate increases for large values of the infection rate (front right in the left plot); this feature cannot be captured by the homogeneous assumption underlying the GP model, and in that region Smoothed model checking underestimates the satisfaction function.

To assess quantitatively the accuracy of the reconstruction, we repeated the simulation for ten different batches of data (keeping both the grid and the hyperparameters of the GP fixed). On average, the Smoothed model checking estimate (posterior GP mean) differed from the deep SMC estimation by 3.7\%, with the deep SMC estimate exceeding the 95\% confidence interval in 4.7\% (again, a remarkably close number to the expected 5\%). The whole Smoothed Model Checking procedure (generating the observation plus GP inference) took on average 1.5 seconds per run on a standard laptop; by contrast, the deep SMC estimates used for the right panel of Figure \ref{fig:SIRdouble} took approximately 13 minutes.

\begin{figure}
\begin{center}
\includegraphics[width=.48\textwidth]{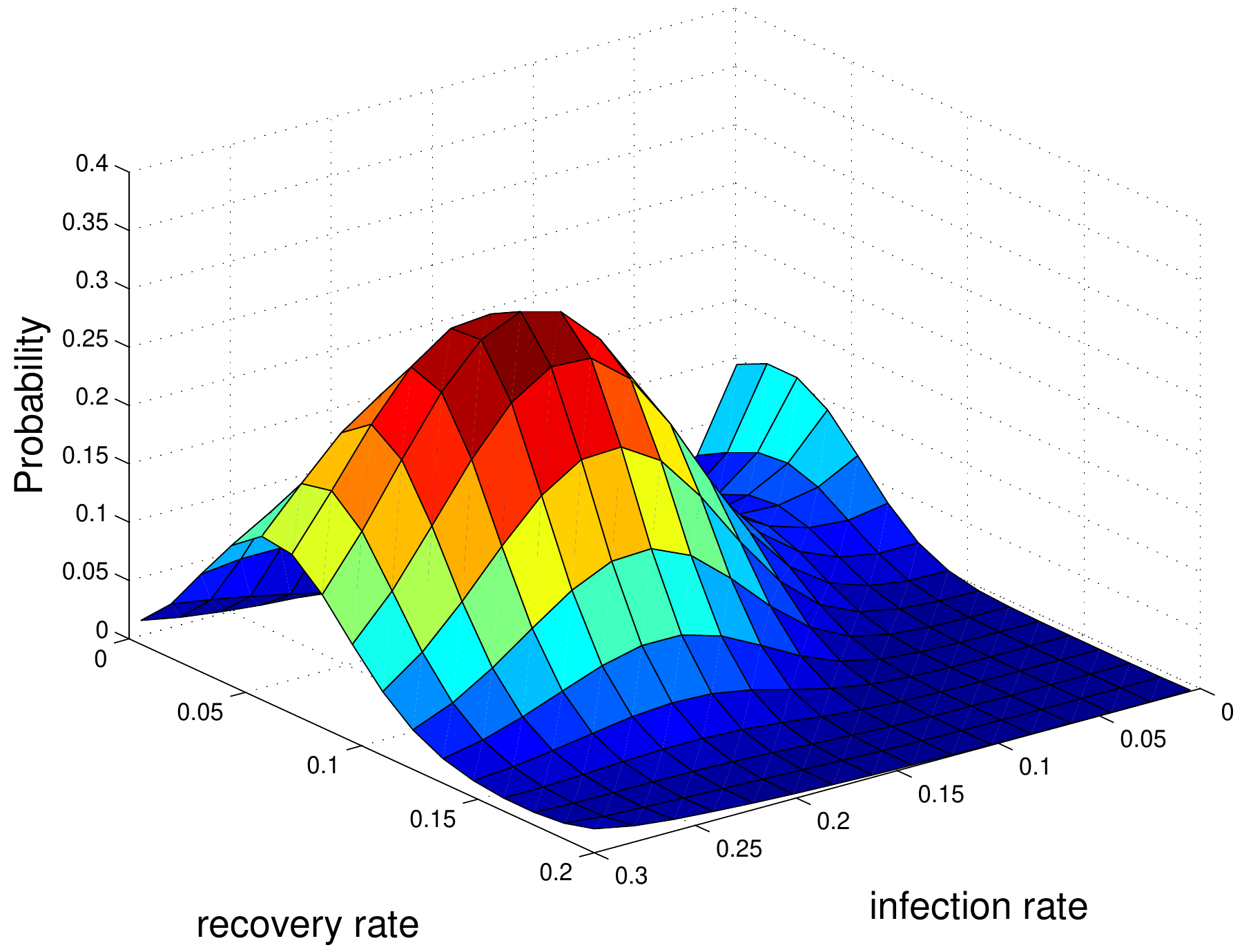}
\includegraphics[width=.48\textwidth]{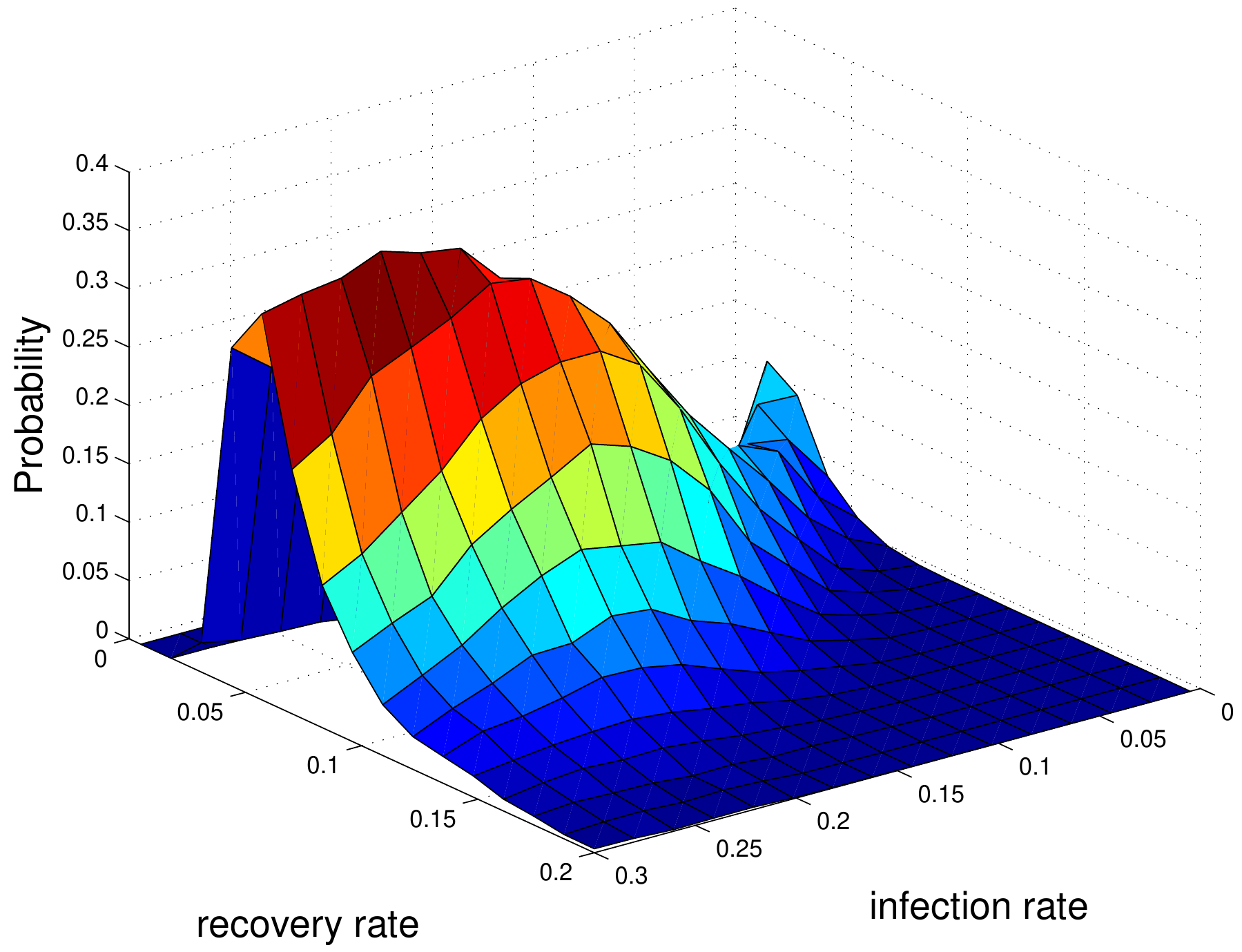}
\end{center}

\caption{SIR model: {\it left}, smoothed model checking reconstruction of the satisfaction probability of the extinction formula (\ref{eq:ext}), as a function of both $k_i\in [0.005,0.3]$  and  $k_r\in[0.005,0.2]$; {\it right} surface plot of a deep statistical model checking estimation over a 16$\times$ 16 grid (5000 simulations per grid point).}
\label{fig:SIRdouble}
\end{figure}

\subsection{Prokaryotic Gene Expression}\label{lacz}

In this section, we consider a more complex model that highlights the computational benefits of smoothed model checking.
That is the model of prokaryotic gene expression \cite{Kierzek2002}, which captures LacZ protein synthesis in E.\ coli.
We perform a series of experiments that demonstrate that the expression of LacZ is affected by varying the transcription and translation initiation frequencies.

This system is modelled as a population CTMC, in which the state space is described by a vector $\vec{X}$ that counts the molecular populations of the species considered.
The dynamics of the CTMC are described by a list of transitions rules, or reactions, together with their rate functions.
We represent them in the biochemical notation style (see e.g. \cite{Gillespie1977}).
All rates of this model follow the law of mass action.  

\begin{figure}[!ht]
\begin{center}
\includegraphics[width=0.8\textwidth]{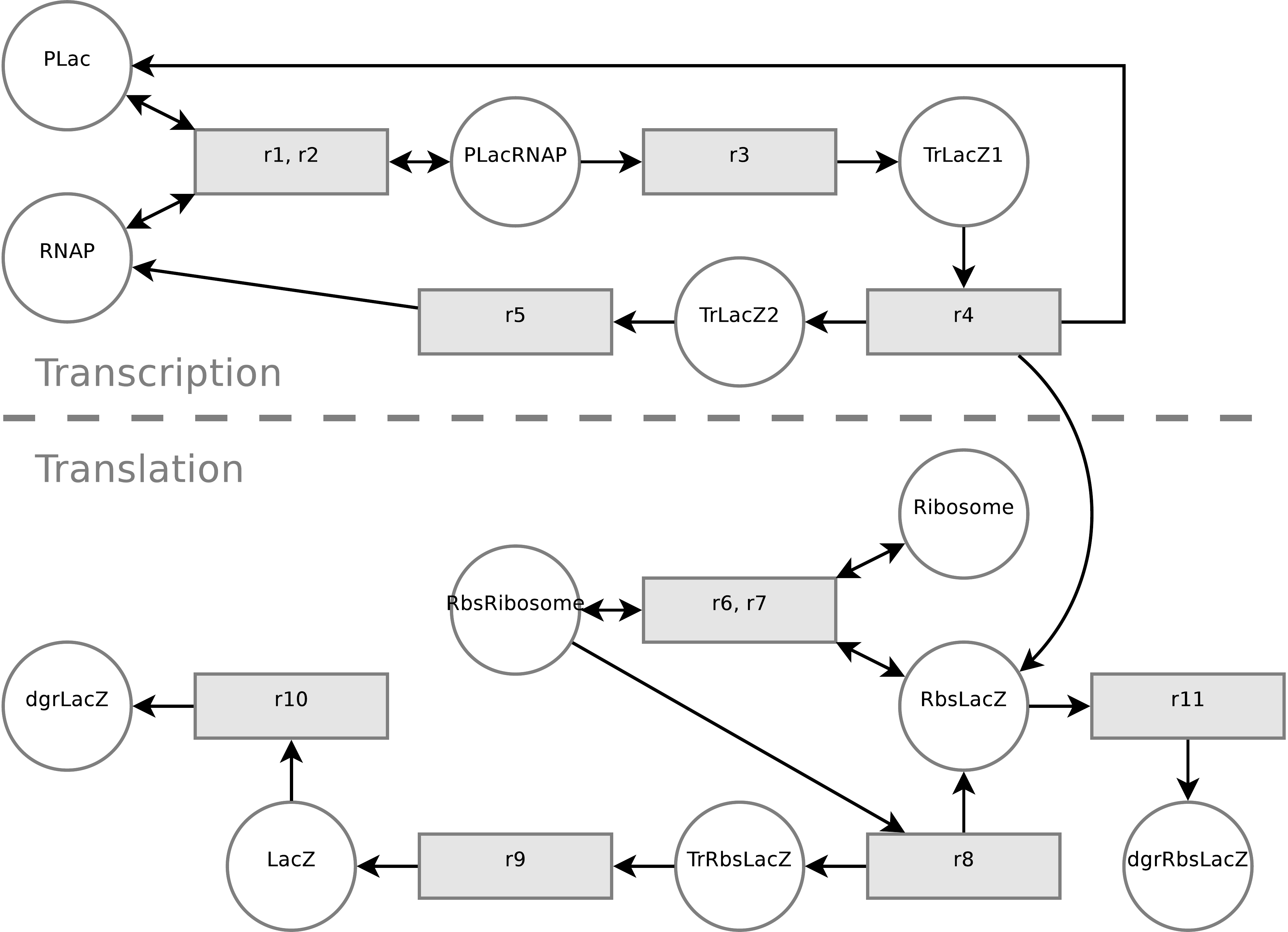}
\end{center}
\caption{Schematic representation of LacZ expression \label{fig:lacz_figure}}
\end{figure}

\begin{table}
\centering
\begin{small}
\renewcommand{\arraystretch}{1.2}
\caption{Rate functions and default parameter values for the LacZ model.}
\label{tab:laczRates}
\begin{tabular}{|r|c|c|}
\hline
Reaction 	& Rate Function						& Default Parameter \\
\hline
\hline
$r_1$		& $k_1 X_{\mathrm{PLac}} X_{\mathrm{RNAP}}$		& $k_1 = 0.17$ \\
$r_2$		& $k_2 X_{\mathrm{PLacRNAP}}$				& $k_2 = 10$ \\
$r_3$		& $k_3 X_{\mathrm{PLacRNAP}}$				& $k_3 = 1$ \\
$r_4$		& $k_4 X_{\mathrm{TrLacZ1}}$				& $k_4 = 1$ \\
$r_5$		& $k_5 X_{\mathrm{TrLacZ2}}$				& $k_5 = 0.015$ \\
$r_6$		& $k_6 X_{\mathrm{Ribosome}} X_{\mathrm{RbsLacZ}}$	& $k_6 = 0.17$ \\
$r_7$		& $k_7 X_{\mathrm{RbsRibosome}}$			& $k_7 = 0.45$ \\
$r_8$		& $k_8 X_{\mathrm{RbsRibosome}}$			& $k_8 = 0.4$ \\
$r_9$		& $k_9 X_{\mathrm{TrRbsLacZ}}$				& $k_9 = 0.015$ \\
$r_{10}$	& $k_{10} X_{\mathrm{LacZ}}$				& $k_{10} = 6.42e-5$ \\
$r_{11}$	& $k_{11} X_{\mathrm{RbsLacZ}}$				& $k_{11} = 0.3$ \\
\hline
\end{tabular}
\end{small}
\end{table}

A schematic representation of the model is depicted in Figure \ref{fig:lacz_figure}.
The rate functions and the default values for the rate constants are summarised in Table \ref{tab:laczRates}.
The first two reactions describe the binding and dissociation of RNA polymerase ($\mathrm{RNAP}$) with the promoter ($\mathrm{PLac}$).
\begin{description}
\item[$r_1$:] $\mathrm{PLac} + \mathrm{RNAP} \xrightarrow{k_1} \mathrm{PLacRNAP}$
\item[$r_2$:] $\mathrm{PLacRNAP} \xrightarrow{k_2} \mathrm{PLac} + \mathrm{RNAP}$
\end{description}
The reactions that follow model RNA transcription, and the release of the promoter and the RNA polymerase.
The product of these reactions is the ribosome binding site for LacZ ($\mathrm{RbsLacZ}$).
The transcription initiation frequency effectively depends on the promoter strength; in the experiments, this is decreased by increasing the RNA polymerase dissociation rate ($k_2$ kinetic constant).
\begin{description}
\item[$r_3$:] $\mathrm{PLacRNAP} \xrightarrow{k_3} \mathrm{TrLacZ1}$
\item[$r_4$:] $\mathrm{TrLacZ1} \xrightarrow{k_4} \mathrm{RbsLacZ} + \mathrm{PLac} + \mathrm{TrLacZ2}$
\item[$r_5$:] $\mathrm{TrLacZ2} \xrightarrow{k_5} \mathrm{RNAP}$
\end{description}
The attachment of a ribosome to a binding site, and the corresponding dissociation, are given as follows:
\begin{description}
\item[$r_6$:] $\mathrm{Ribosome} + \mathrm{RbsLacZ} \xrightarrow{k_6} \mathrm{RbsRibosome}$
\item[$r_7$:] $\mathrm{RbsRibosome} \xrightarrow{k_7} \mathrm{Ribosome} + \mathrm{RbsLacZ}$
\end{description}
The following reactions describe the last two steps in gene expression: translation and protein synthesis.
In the experiments, the translation initiation frequency is decreased by increasing the ribosome dissociation rate ($k_7$ kinetic constant).
\begin{description}
\item[$r_8$:] $\mathrm{RbsRibosome} \xrightarrow{k_8} \mathrm{TrRbsLacZ} + \mathrm{RbsLacZ}$
\item[$r_9$:] $\mathrm{TrRbsLacZ} \xrightarrow{k_9} \mathrm{LacZ}$
\end{description}
Finally, the reactions that follow model the degradation of the LacZ protein and the transcribed RNA correspondingly.
\begin{description}
\item[$r_{10}$:] $\mathrm{LacZ} \xrightarrow{k_{10}} \mathrm{dgrLacZ}$
\item[$r_{11}$:] $\mathrm{RbsLacZ} \xrightarrow{k_{11}} \mathrm{dgrRbsLacZ}$
\end{description}

\subsubsection{Measuring Stochastic Fluctuations}

The authors of the original paper \cite{Kierzek2002} explored how the transcription and translation initiation frequencies affect the stochastic fluctuations in gene expression via experimentation with a range of parameter values.
We shall repeat this experiment in a slightly different setting, where enquiries regarding stochastic variation are expressed in terms of temporal logic, and results are obtained via smoothed model checking, rather than na\"{i}ve exploration of the parameter space.

In this work, we shall measure stochastic variation using the following MiTL property, which states that at least at one interval of $5000$ seconds, the difference of the LacZ population is always smaller than $10$\% of the expected value:
\begin{equation}
\label{eq:variationFormula}
\phi =  \eventually{0}{21000} \left(\always{0}{5000} \left|X_{\mathrm{LacZ}} - E[X_{\mathrm{LacZ}}]\right| <  0.1 \times E[X_{\mathrm{LacZ}}]\right)
\end{equation}
The expectation $E[X_{\mathrm{LacZ}}]$ denotes the expanded value of LacZ over the course of time.
In our experiments, these expectations have been approximated by statistical means.

We first used smoothed model checking to explore each of the two parameters individually.
We varied the RNA polymerase dissociation rate $k_2$ in the interval $k_2 \in [10,100000]$, holding $k_7$ fixed to $k_7 = 0.45$, while $k_7$ has been varied within $[0.45,4500]$ with $k_1=10$. 
In each case, we sampled 10 runs from a grid of 25 points evenly distributed in the parameter domain, for a total of $250$ runs, and applied smoothed model checking to obtain a prediction of the satisfaction probability of $\phi$ as a function of the parameter on $100$ points.

Results are shown in Figure \ref{fig:lacz_variation_signle}, where we also plot the 95\% confidence bounds of the estimate.
The predicted satisfaction function is compared with estimates of the satisfaction probability at the $25$ values used for training the Gaussian process using standard statistical model checking with $100$ runs per point.
The running times for both standard and smoothed MC are summarised in Table \ref{tab:timesVariation}.
The smoothed MC approach has been broken down into three steps: the initial SMC estimations, the hyperparameter optimisation for the GP, and the GP prediction over the specified grid.

\begin{figure}[htp]
\includegraphics[width=0.48\textwidth]{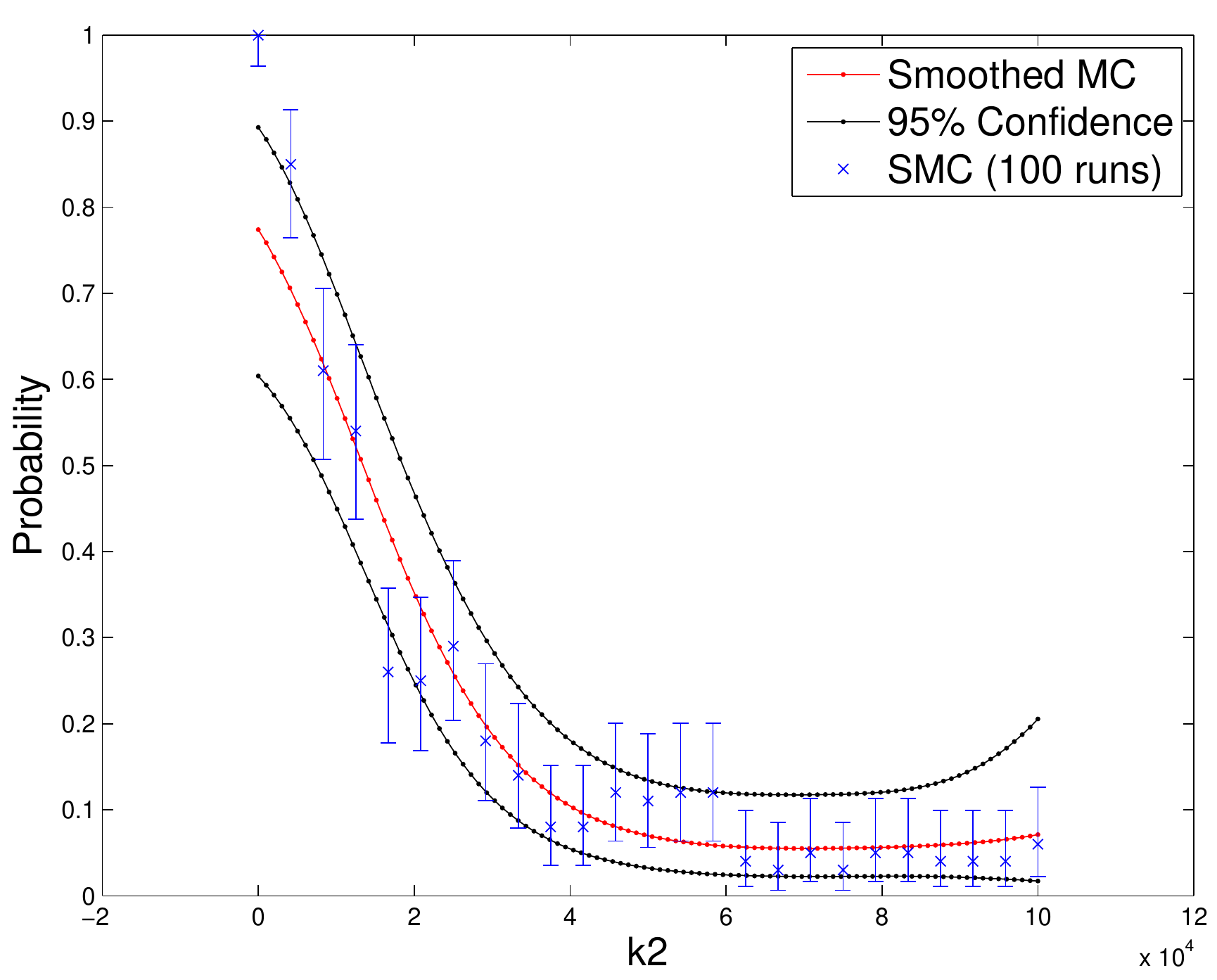}
\includegraphics[width=0.48\textwidth]{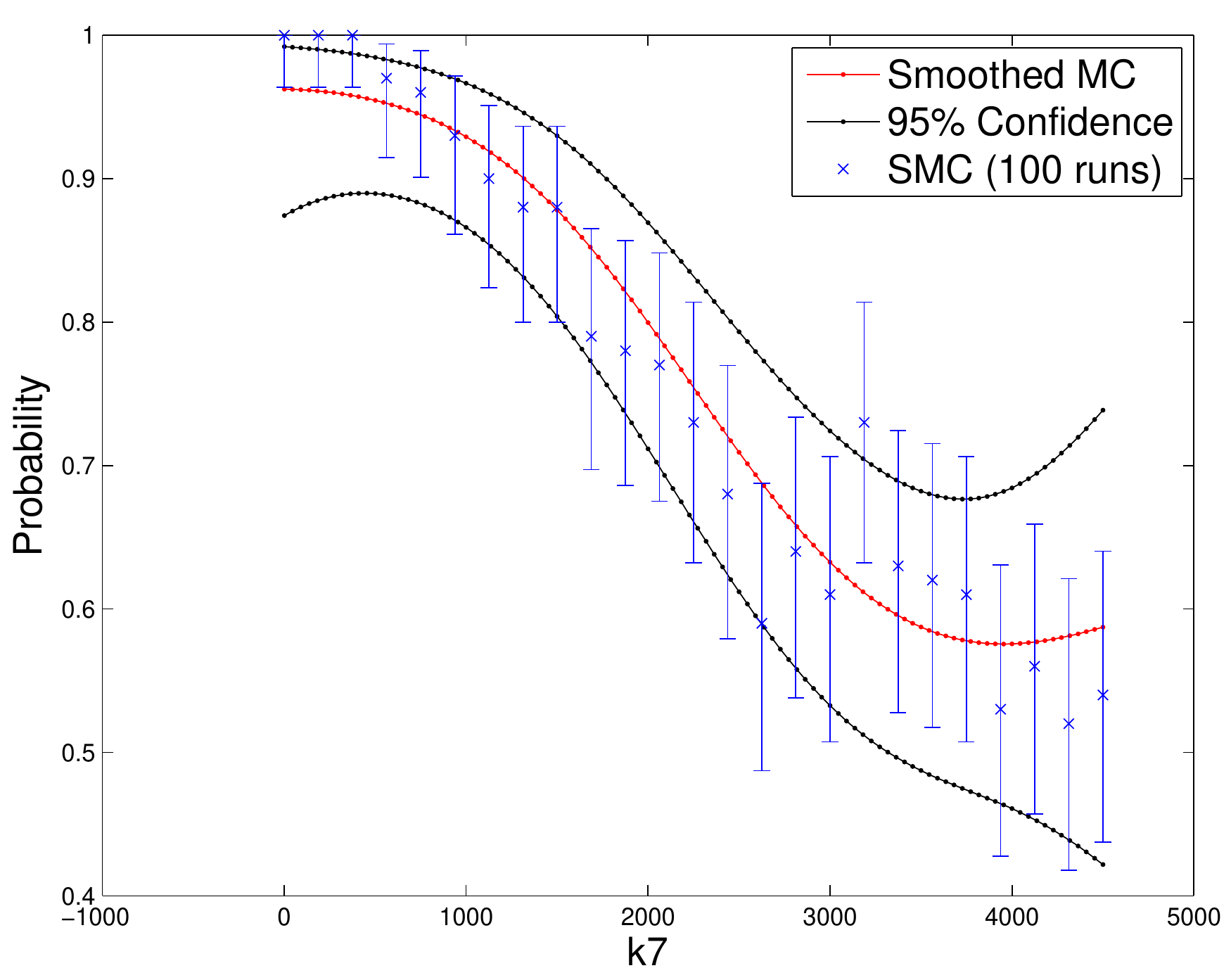}
\caption{Smoothed model checking reconstruction of the satisfaction probability of the variation formula in (\ref{eq:variationFormula}), as a function of $k_2 \in [10,100000]$ (left) and $k_7 \in [0.45,4500]$ (right).}
\label{fig:lacz_variation_signle}
\end{figure}

\begin{table}[!t]
\centering
\renewcommand{\arraystretch}{1.2}
\caption{Running times of model checking the variation formula in (\ref{eq:variationFormula}) for the LacZ model. The experiments have been performed in an Intel\textregistered\ Xeon\texttrademark\ E5410 @ 2.33GHz PC running Linux.
}
\label{tab:timesVariation}
\begin{threeparttable}[b]
\begin{tabular}{|c|r||>{\centering}m{0.16\textwidth}>{\centering}m{0.16\textwidth}c|}
\hline
\multicolumn{2}{|c||}{Method} 			&$k_2$ 			&$k_7$			&$k_2, k_7$ \\
\hline
\hline
\multirow{4}{*}{Smoothed MC} & SMC (10 runs)	& 32 sec\tnote{*}	& 92 sec\tnote{*}	& 160 sec\tnote{**}\ \ \ \\
			     & Hyperparam.\ Opt.& 3 sec 		& 3 sec			& 17 sec \\
			     & GP Prediction    & 0.05 sec\tnote{\dag} 	& 0.05 sec\tnote{\dag}	& 0.3 sec\tnote{\ddag} \\
			     & Total 		& 35 sec 	  	& 95 sec 		& 177 sec \\
\hline
\multicolumn{2}{|c||}{Na\"{i}ve SMC (100 runs)} & 320 sec\tnote{*}	& 920 sec\tnote{*}	& 1600 sec\tnote{**}\ \ \ \  \\
\hline
\end{tabular}
\begin{small}
\begin{tablenotes}
\begin{minipage}{0.5\linewidth}
	\item[*] Grid of $25$ parameter values.
	\item[**] Grid of $100$ parameter values.
\end{minipage}%
\begin{minipage}{0.5\linewidth}
	\item[\dag] Grid of $100$ parameter values.
	\item[\ddag] Grid of $400$ parameter values.
\end{minipage}
\end{tablenotes}
\end{small}
\end{threeparttable}
\end{table}

In the second experiment, we vary both $k_2 \in [10,10000]$ and $k_7 \in [0.45,450]$ simultaneously.
We have sampled $100$ regularly distributed parameter values, for each of which we run $10$ simulation runs.
The estimated satisfaction function over $400$ points is plotted in Figure \ref{fig:lacz_variation_double}.
On the right-hand side of Figure \ref{fig:lacz_variation_double}, we also show the estimated probabilities over $100$ values via standard model checking using $100$ simulation runs.
Both methods are in agreement regarding the shape of the satisfaction function.
With smoothed model checking however, we obtain a more refined grid of estimations in a fraction of the time required for na\"{i}ve parameter exploration, as we also see in Table \ref{tab:timesVariation}.

\begin{figure}[htp]
\includegraphics[width=0.48\textwidth]{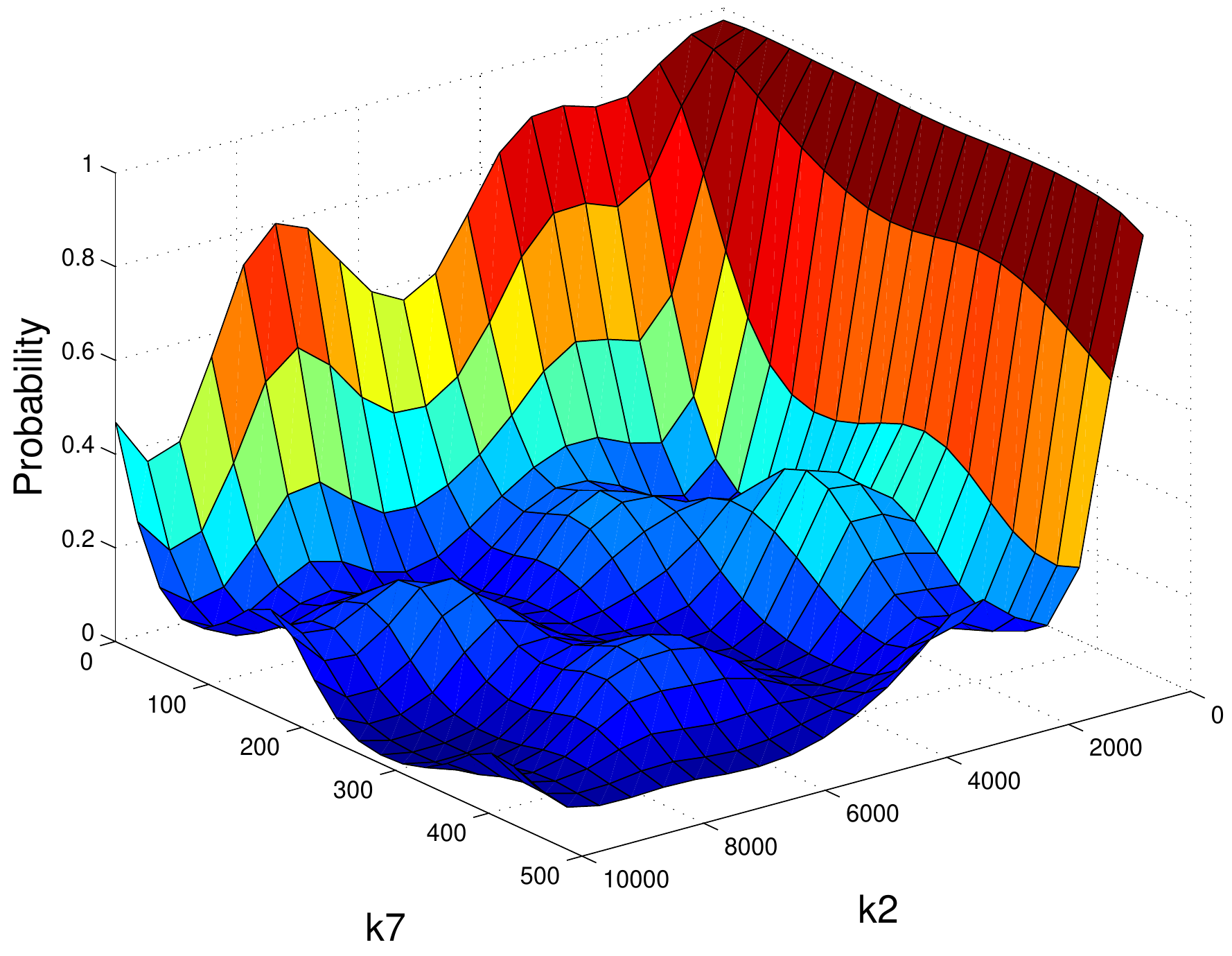}
\includegraphics[width=0.48\textwidth]{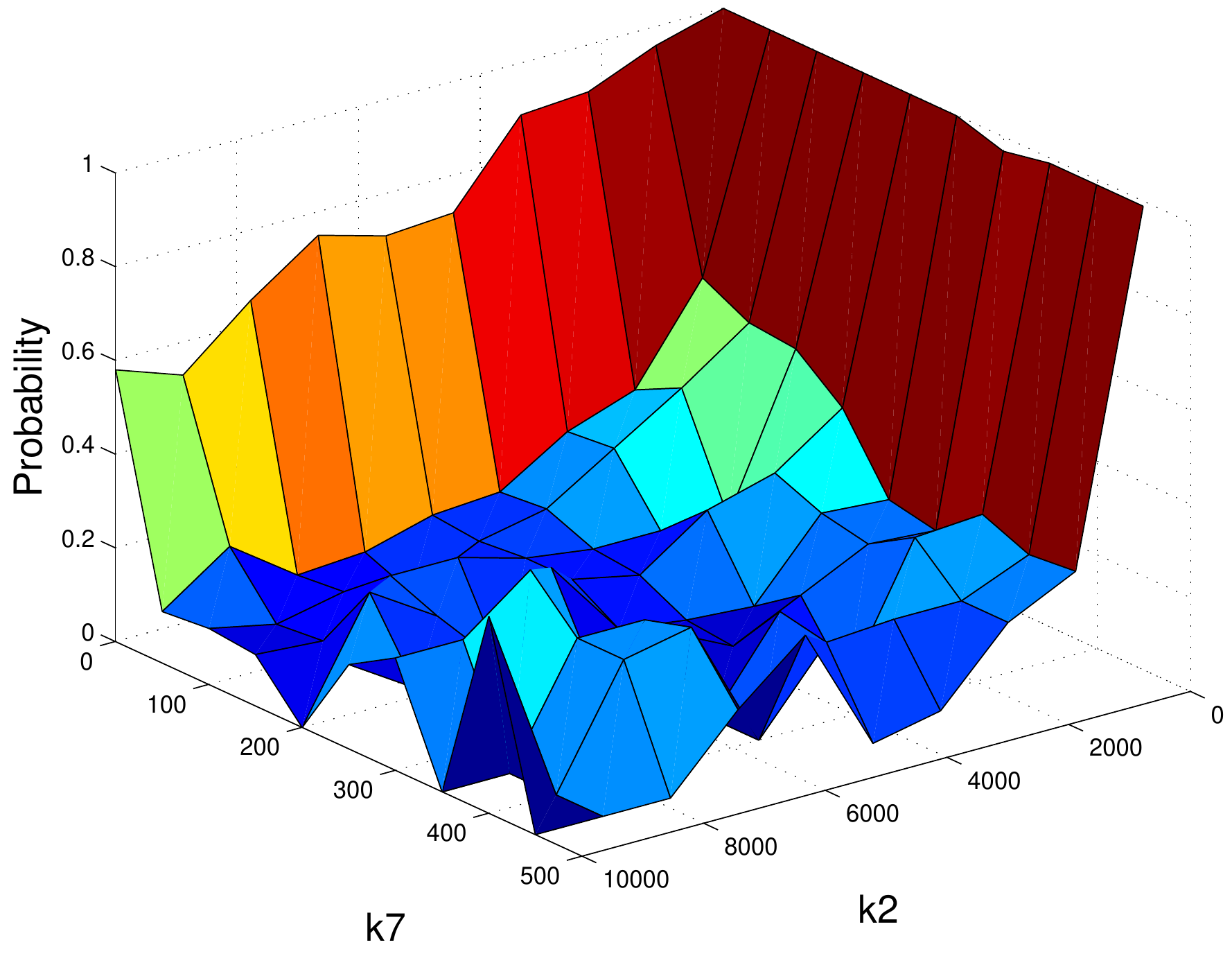}
\caption{Left: Smoothed MC reconstruction of the satisfaction probability of the variation formula in (\ref{eq:variationFormula}), as a function of $k_2 \in [10,10000]$ and $k_7 \in [0.45,450]$. Right: the same satisfaction probability function as estimated by na\"{i}ve SMC over a coarser grid of parameter values.}
\label{fig:lacz_variation_double}
\end{figure}

As a final remark on measuring stochastic fluctuations in gene expression, the satisfaction probability of $\phi$, which denotes whether a trajectory deviates more than 10\% from the average, decreases as the dissociation rates for RNAP and Ribosome are increased.
In other words, decreasing the transcription and translation initiation frequencies introduces more stochastic variation in the expression of LacZ.
The same conclusion has been drawn by Kierzek \cite{Kierzek2002} by measuring the variation coefficient for the population of LacZ.

\subsubsection{Detecting Bursts of Gene Expression}

As reported in \cite{Kierzek2002}, decreasing the transcription initiation rate results in a behaviour that is characterised by irregular \lq\lq bursts\rq\rq\ of gene expression.
We formalise the concept of burst as a MiTL formula which monitors rapid increases in LacZ counts, followed by long periods of lack of protein production. The resulting formula is 
\begin{equation}
\label{eq:burstFormula}
\phi =  \eventually{16000}{21000} \left( \Delta X_{\mathrm{LacZ}} >  0 \; \wedge \; \always{10}{2000} \Delta X_{\mathrm{LacZ}} \leq  0 \right)
\end{equation}
which will be true if a burst of gene expression is present in a particular time-window, from $16000$ to $21000$ seconds.

The parameters explored are the RNA polymerase dissociation rate $k_2$, and the ribosome dissociation rate $k_7$.
Figure \ref{fig:lacz_burst_signle} shows the results of smoothed MC over each of the two parameters individually.
The $k_2$ parameter has been varied in the interval $[10,10000]$, holding $k_7$ fixed to $k_7 = 0.45$, while $k_7$ has been varied within $[0.45,450]$ with $k_1=10$. 
For each experiment, we have sampled $25$ regularly distributed parameter values, and we have queried for the satisfaction probability at $100$ parameter values.
We also plot the 95\% confidence bounds of the estimate.

The predicted satisfaction function is compared with estimates of the satisfaction probability at the $25$ points using standard SMC with $100$ runs per point.
These are also plotted in Figure \ref{fig:lacz_burst_signle} as points with error bars.
We see that the majority of na\"{i}ve SMC estimations lie within the 95\% confidence intervals of the GP.
Most importantly, our approach has been significantly more efficient.
Table \ref{tab:timesBurst} summarises the running times for smoothed MC, which have bee broken down to initial SMC estimations, hyperparameter optimisation and GP regression.
The total time required is only a fraction of that of the standard approach.


\begin{table}[!t]
\centering
\renewcommand{\arraystretch}{1.2}
\caption{Running times of model checking the expression burst formula in (\ref{eq:burstFormula}) for the LacZ model. The experiments have been performed in an Intel\textregistered\ Xeon\texttrademark\ E5410 @ 2.33GHz PC running Linux.}
\label{tab:timesBurst}
\begin{threeparttable}
\begin{tabular}{|c|r||>{\centering}m{0.16\textwidth}>{\centering}m{0.16\textwidth}c|}
\hline
\multicolumn{2}{|c||}{Method} 			&$k_2$ 			&$k_7$			&$k_2, k_7$ \\
\hline
\hline
\multirow{4}{*}{Smoothed MC} & SMC (10 runs)	& 16 sec\tnote{*} 	& 34 sec\tnote{*}	& 77 sec\tnote{**} \\
			     & Hyperparam.\ Opt.& 2 sec 	  	& 2 sec			& 23 sec \\
			     & GP Prediction    & 0.07 sec\tnote{\dag}	& 0.06 sec\tnote{\dag}	& 0.28 sec\tnote{\ddag} \\
			     & Total 		& 18 sec 	 	& 36 sec 		& 100 sec \\
\hline
\multicolumn{2}{|c||}{Na\"{i}ve SMC (100 runs)} & 165 sec\tnote{*}\ \ 	& 346 sec\tnote{*}\ \ \	& 772 sec\tnote{**}\ \ \ \ \\
\hline
\end{tabular}
\begin{small}
\begin{tablenotes}
\begin{minipage}{0.5\linewidth}
	\item[*] Grid of $25$ parameter values.
	\item[**] Grid of $100$ parameter values.
\end{minipage}
\begin{minipage}{0.5\linewidth}
	\item[\dag] Grid of $100$ parameter values.
	\item[\ddag] Grid of $400$ parameter values.
\end{minipage}
\end{tablenotes}
\end{small}
\end{threeparttable}
\end{table}

\begin{figure}[htp]
\includegraphics[width=0.48\textwidth]{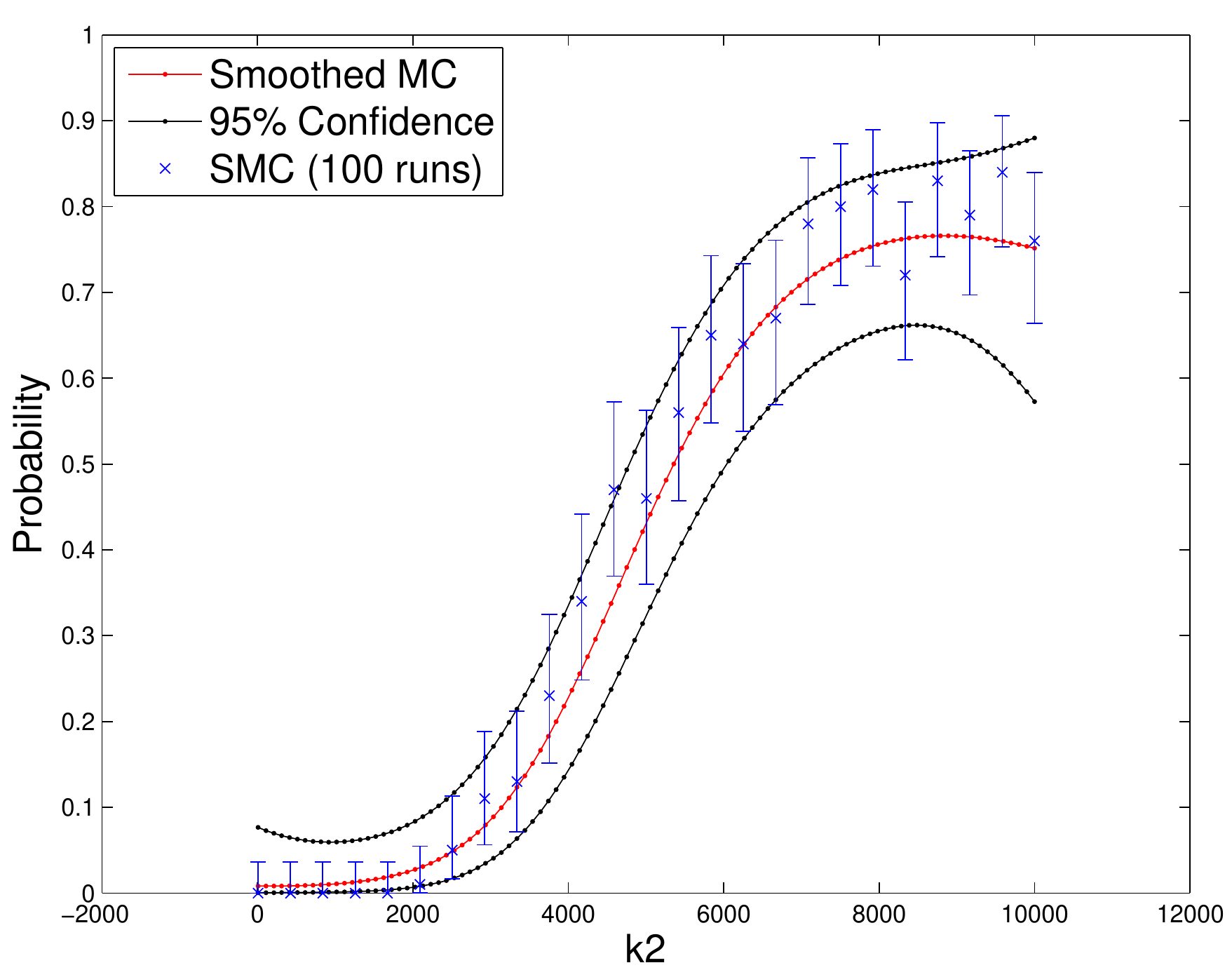}
\includegraphics[width=0.48\textwidth]{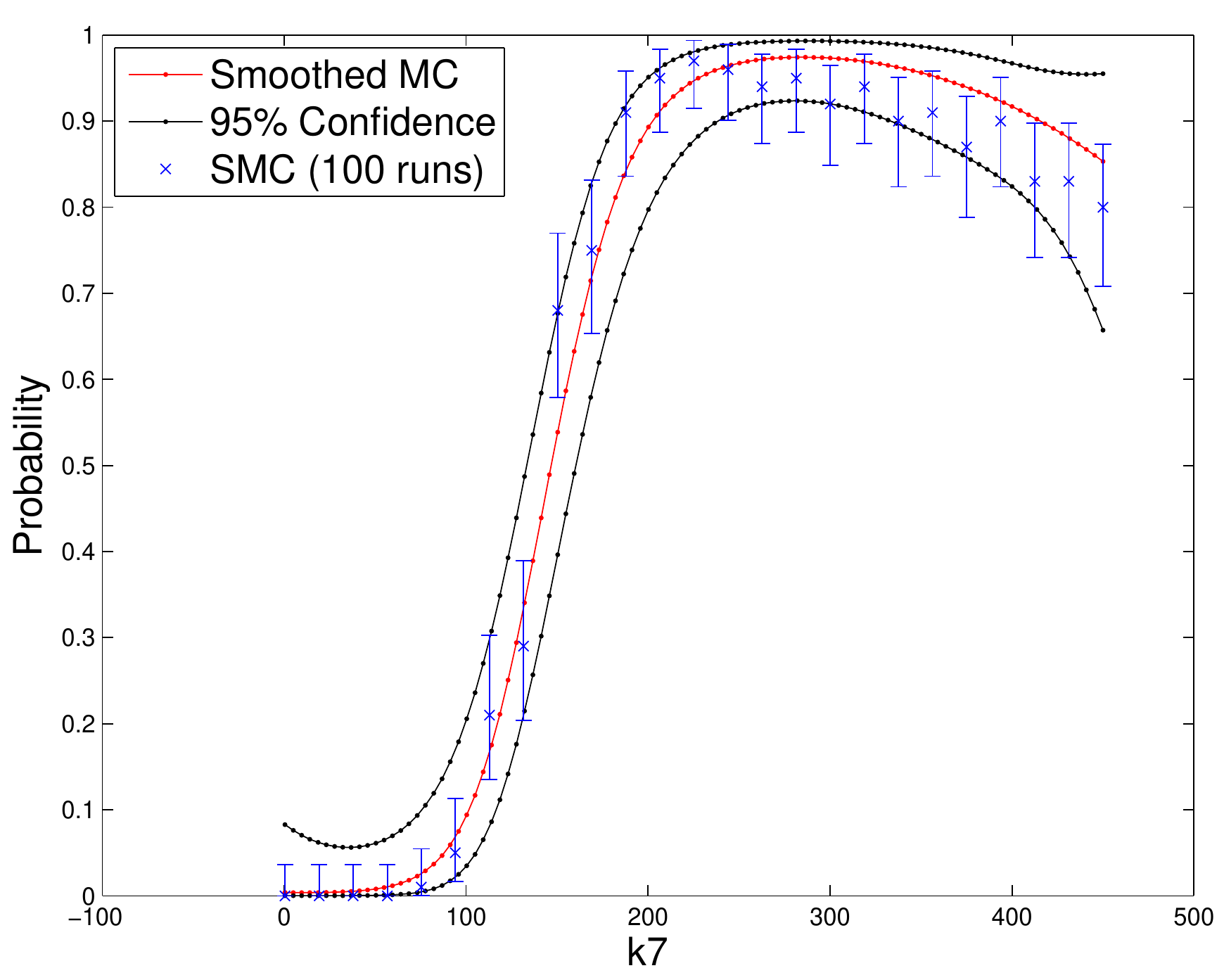}
\caption{Smoothed model checking reconstruction of the satisfaction probability of the expression burst formula in (\ref{eq:burstFormula}), as a function of $k_2 \in [10,10000]$ (left) and $k_7 \in [0.45,450]$ (right).}
\label{fig:lacz_burst_signle}
\end{figure}

In the second experiment, we vary both $k_2 \in [10,10000]$ and $k_7 \in [0.45,450]$ simultaneously.
We have sampled $100$ regularly distributed parameter values, for each of which we run $10$ simulation runs, and we have queried for the satisfaction probability at $400$ parameter values.
The estimated satisfaction function over $400$ points is plotted in Figure \ref{fig:lacz_burst_double}.
On the right-hand side of Figure \ref{fig:lacz_burst_double}, we also show the estimated probabilities over $100$ values via standard model checking using $100$ simulation runs.
Both methods are in agreement regarding the shape of the satisfaction function.
With smoothed model checking however, we obtain a more refined grid of estimations in a fraction of the time required for na\"{i}ve parameter exploration, as seen in Table \ref{tab:timesBurst}.
We used our own implementation of SMC; computational savings are possible using more efficient implementations but these will affect equally SMC and Smoothed MC results, since the SMC runs in Smoothed MC account for over 85\% of its running time.

\begin{figure}[htp]
\includegraphics[width=0.48\textwidth]{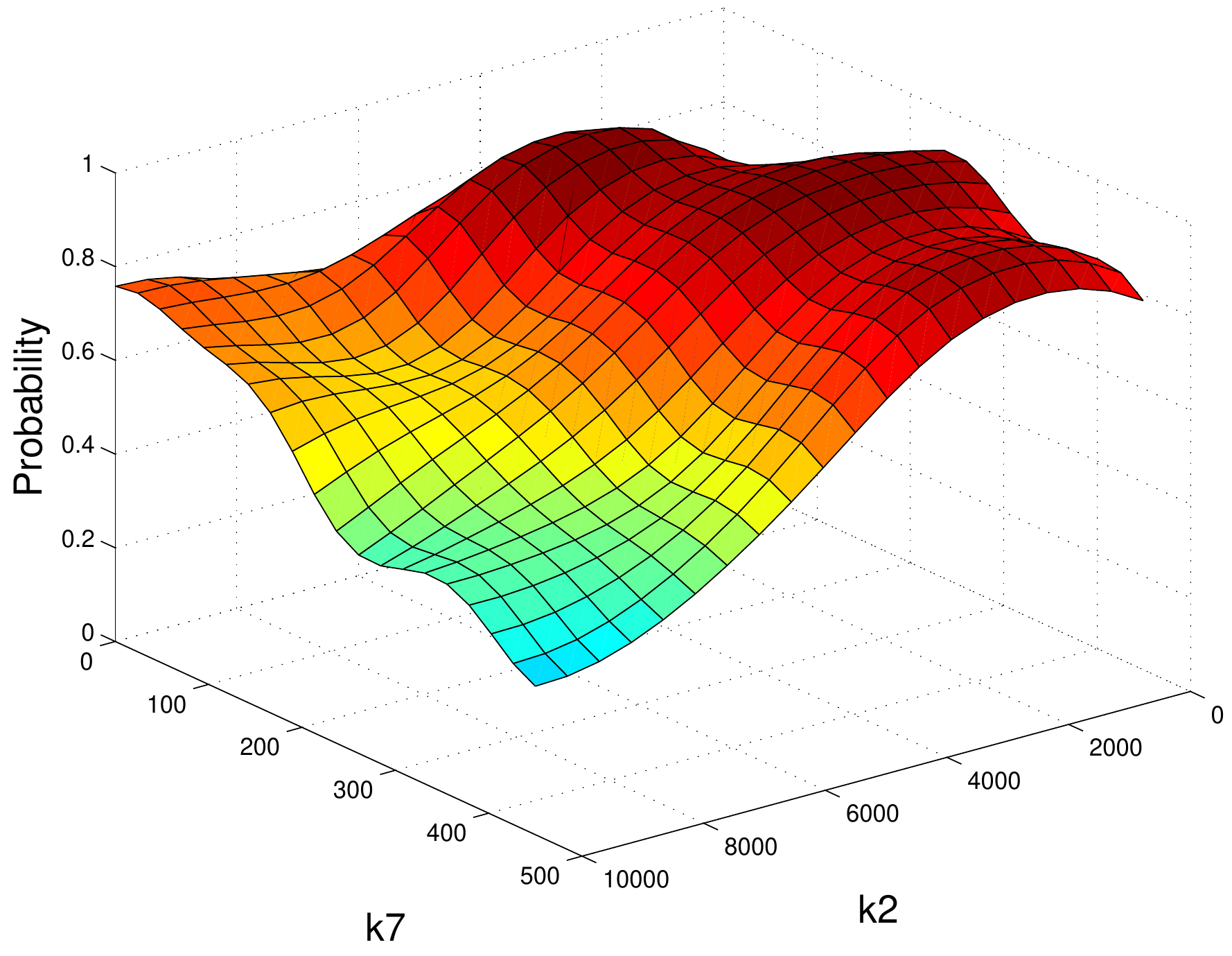}
\includegraphics[width=0.48\textwidth]{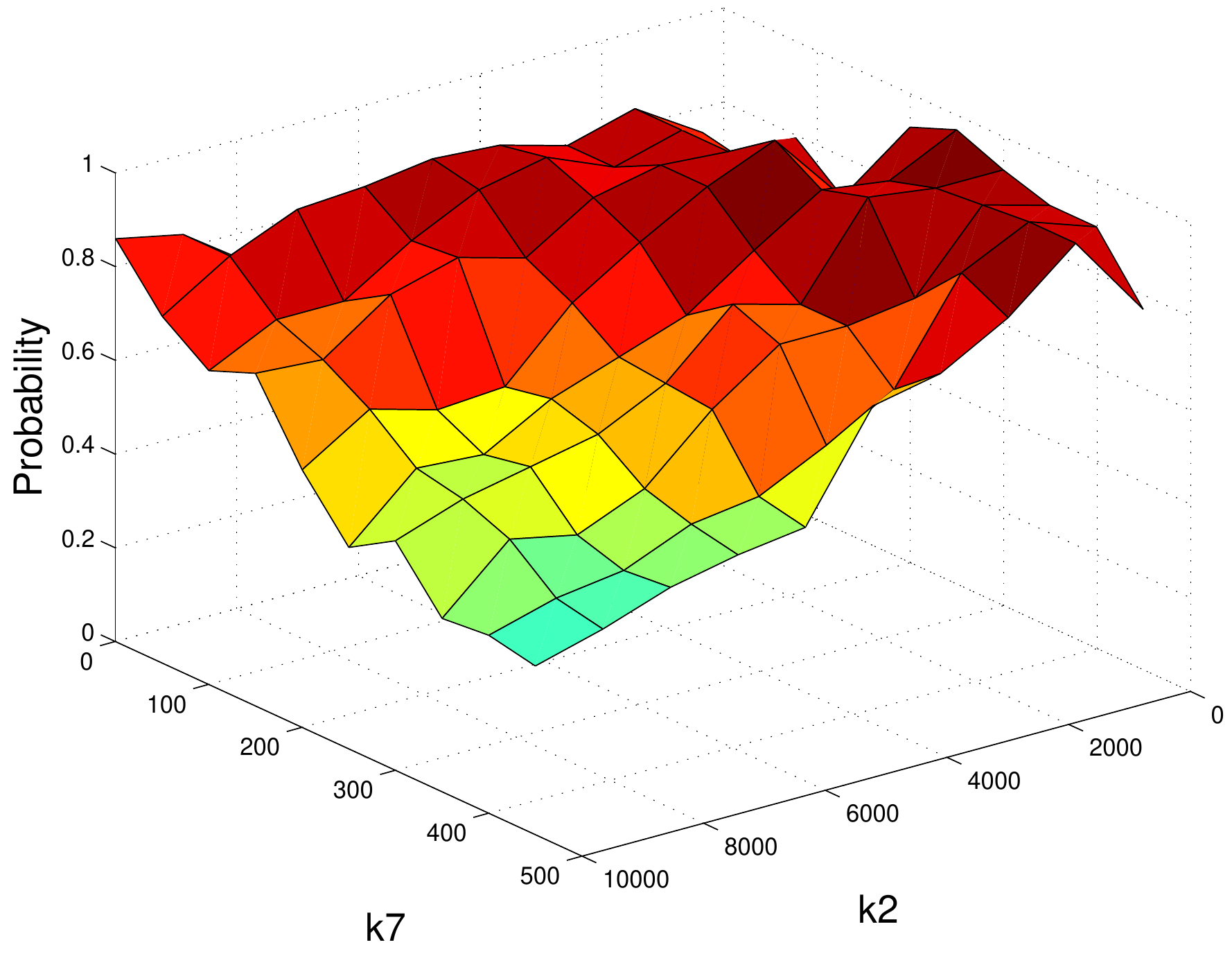}
\caption{Left: smoothed model checking reconstruction of the satisfaction probability of the expression burst formula in (\ref{eq:burstFormula}), as a function of $k_2 \in [10,10000]$ and $k_7 \in [0.45,450]$. Right: the same satisfaction probability function as estimated by na\"{i}ve statistical model checking over a coarser grid of parameter values.}
\label{fig:lacz_burst_double}
\end{figure}


\section{Related work}
Dynamical systems with parametric uncertainty have attracted considerable attention in recent years: some of the most relevant related concepts are Constrained Markov Chains (CMC)\cite{caillaud2011constraint}, Continuous-Time Markov Decision Processes (CTMDP)\cite{Baier2005,Guo}, and  Imprecise CTMC (ICTMC, \cite{DeCooman2009}). Our scenario is somehow simpler than CTMDP and ICTMC, because we assume that there is a true, but unknown, parameter value, which remains constant during the dynamics, while in CTMDPs and ICTMC, instead, the parameter value is allowed to change after every jump. As far as model checking is considered, there have been few attempts to deal with uncertain or imprecise CTMC. In \cite{Katoen2007}, the authors  develop an abstraction framework of standard CTMCs (with known parameters) based on a three valued logic and on ICTMC, and they consider the numerical model checking problem of the reduced ICTMC model for CSL logical properties, exploiting efficient CTMDP algorithms \cite{Baier2005}. The result of such model checking procedure, as far as probabilities of path formulae are concerned, is an interval of probability values. Similar in spirit is the method of \cite{benedikt_ltl_2013}, where authors provide an Expectation-Maximisation algorithm to compute maximum and minimum satisfaction probabilities for LTL formulae.  A statistical model checking scheme for MDPs has been presented in \cite{henriques_statistical_2012}.
More related to this work is the approach for uncertain CTMC and CSL model checking of Brim et al \cite{Brim2013}, in which the authors split the parameter space into small regions and compute an upper and a lower bound on the satisfaction probability in each region. A combination of this approach with statistical sampling around putative maxima for optimisation purposes has been developed in \cite{milan14}.  However, all these approaches are computationally  intensive. As far as we know, our approach is the first one to explicitly leverage the smoothness of the satisfaction function to develop statistical methods to estimate the entire satisfaction function of the probability distribution. 

The use of model checking for system identification (learning parameters from data) was first proposed in \cite{Donaldson2008}, which used an optimisation scheme based on genetic algorithms to select plausible parameter values. More recently, \cite{Bortolussi2013qest} proposed a rigorous statistical framework for property-based system identification (and system design), using a GP-based optimisation method to optimise the likelihood of observed properties, encoded as MiTL formulae. System design in a model checking framework has also been carried out in \cite{Bartocci2013}, using again a GP-based optimisation routine, but exploiting a stochastic version of the  quantitative semantics of MiTL of \cite{donze2010robust}.  Model checking methods within optimisation routines have also been used in the related problems of model repair \cite{Bartocci2011} and of learning parametric formulae describing a system \cite{Bartocci:learning13}. 

%

%
%

\section{Conclusions}
Verification of logical properties over uncertain stochastic processes is an important task in theoretical computer science. This paper offers contributions on two different levels to this model checking problem: from the theoretical point of view, it refocuses the question of model checking from estimating a number (the satisfaction probability of a formula) to estimating a {\it function}, which we term the satisfaction function of the formula. Our main theoretical result characterises the satisfaction function as a differentiable function of the parameters of the process, under mild conditions. From the practical point of view, this smoothness result offers a powerful new way to perform statistical model checking: intuitively, the power of statistical model checking, deriving from the law of large numbers, can be increased by simultaneous sampling at nearby values of the parameters. We show that Gaussian Processes provide a natural prior distribution over smooth functions which can be employed in a Bayesian statistical model checking framework to evaluate a whole satisfaction function from a relatively small number of samples. We term this novel approach to model checking Smoothed Model Checking, and show in an empirical section that indeed this approach can be extremely efficient and accurate.

The availability of quick methods for estimating a satisfaction function could be of considerable use in tackling other computational problems: for example, such a method could be used to inform system design approaches, or to guide effectively and efficiently importance sampling strategies for SMC \cite{jegourel_cross-entropy_2012,zuliani_rare-event_2012} to identify combinations of parameters giving large probability to the rare event we wish to estimate. We stress, however, that the purpose of importance sampling methods is somewhat complementary to our aim: while importance sampling focuses on producing accurate estimates of rare event probabilities at specific values of the parameters, our method provides a global analytical approximation to the satisfaction function. To our knowledge, such global approximations have not been studied so far, and could represent a novel direction in statistical model checking. 

The cross fertilisation of ideas from machine learning and  model checking is increasingly being exploited in the development of novel algorithms in both fields. GP-based methods in particular are increasingly becoming part of the formal modelling repertoire \cite{Bortolussi2013qest,Bartocci:data14,Legay:statistical14}. A practical limitation shared by all GP-based methods is the cubic complexity of the matrix inversion operations needed in GP prediction. Scaling our method to large spaces of parameters would therefore necessitate the use of approximate methods to alleviate the computational burden; we envisage that ideas from sparse GP prediction could be very effective in this scenario \cite{Rasmussen2006}.




\clearpage

\bibliographystyle{plain}
\bibliography{../TACAS15/TACAS15}

\end{document}